\documentclass[draftclsnofoot,onecolumn,12pt]{IEEEtran}

\usepackage{graphicx,epsf,psfrag,amsmath,amssymb,amsfonts,verbatim}
\usepackage[mathscr]{eucal}

\newtheorem{theorem}{Theorem}
\newtheorem{corollary}{Corollary}
\newtheorem{lemma}{Lemma}

\newtheorem{prop}{Proposition}
\newtheorem{property}{Property}

\newcommand{\Cmsc}{\mathscr{C}}

\newcommand{\Emsc}{\mathscr{E}}

\newcommand{\Lmsc}{\mathscr{L}}

\newcommand{\Nmsc}{\mathscr{N}}

\newcommand{\Pmsc}{\mathscr{P}}

\newcommand{\Umsc}{\mathscr{U}}
\newcommand{\Vmsc}{\mathscr{V}}
\newcommand{\Wmsc}{\mathscr{W}}
\newcommand{\Xmsc}{\mathscr{X}}

\newcommand{\Ocal}{\mathcal{O}}

\newcommand{\Rmbb}{\mathbb{R}}
\newcommand{\Embb}{\mathbb{E}}
\newcommand{\Fmbb}{\mathbb{F}}
\newcommand{\Zmbb}{\mathbb{Z}}

\newcommand{\xbf}{\mathbf{x}}
\newcommand{\Xbf}{\mathbf{X}}

\newcommand{\wbf}{\mathbf{w}}
\newcommand{\Wbf}{\mathbf{W}}
\newcommand{\Zbf}{\mathbf{Z}}

\newcommand{\tA}{\tilde{A}}

\newcommand{\tC}{\tilde{C}}
\newcommand{\tD}{\tilde{D}}

\newcommand{\tF}{\tilde{F}}
\newcommand{\tX}{\tilde{X}}
\newcommand{\tY}{\tilde{Y}}
\newcommand{\tZ}{\tilde{Z}}
\newcommand{\tXbf}{\tilde{\Xbf}}
\newcommand{\tZbf}{\tilde{\Zbf}}
\newcommand{\tWbf}{\tilde{\Wbf}}

\newcommand{\beq}{\begin{equation}}
\newcommand{\eeq}{\end{equation}}
\newcommand{\eps}{\epsilon}

\newcommand{\bi}{\begin{itemize}}
\newcommand{\ei}{\end{itemize}}

\newcommand{\rank}{\text{rank}}

\newcommand{\head}{\text{head}}
\newcommand{\tail}{\text{tail}}
\newcommand{\inn}{\text{in}}
\newcommand{\out}{\text{out}}
\newcommand{\spn}{\text{span}}

\title{Polytope Codes Against Adversaries in Networks}

\author{Oliver Kosut, Lang Tong, and David Tse%
\thanks{O. Kosut is with the Massachusetts Institute of Technology, Cambridge, MA {\tt okosut@mit.edu}}
\thanks{L. Tong is with Cornell University, Ithaca, NY {\tt ltong@ece.cornell.edu}}
\thanks{D. Tse is with the University of California, Berkeley, CA {\tt dtse@eecs.berkeley.edu}}
\thanks{This work is supported in part by
the National Science Foundation under Award CCF-0635070
and the Army Research Office under Grant ARO-W911NF-06-1-0346.}}

\begin{document}

\maketitle

\begin{abstract}
Network coding is studied when an adversary controls a subset of nodes in the network of limited quantity but unknown location. This problem is shown to be more difficult than when the adversary controls a given number of edges in the network, in that linear codes are insufficient. To solve the node problem, the class of Polytope Codes is introduced. Polytope Codes are constant composition codes operating over bounded polytopes in integer vector fields. The polytope structure creates additional complexity, but it induces properties on marginal distributions of code vectors so that validities of codewords can be checked by internal nodes of the network. It is shown that Polytope Codes achieve a cut-set bound for a class of planar networks. It is also shown that this cut-set bound is not always tight, and a tighter bound is given for an example network.
\end{abstract}

\section{Introduction}

Network coding allows routers in a network to execute possibly complex codes in addition to routing; it has been shown that allowing them to do so can increase communication rate \cite{AhlswedeEtal:00IT}. However, taking advantage of this coding at internal nodes means that the sources and destinations must rely on other nodes---nodes they may not have complete control over---to reliably perform certain functions. If these internal nodes do not behave correctly, or, worse, maliciously attempt to subvert the goals of the users---constituting a so-called Byzantine attack \cite{LamportShostakPease:82ACM,Dolev:82}---standard network coding techniques fail.

Suppose an omniscient adversary controls an unknown portion of the network, and may arbitrarily corrupt the transmissions on certain communication links. We wish to determine how the size of the adversarial part of the network influences reliable communication rates. If the adversary may control any $z$ unit-capacity edges in the network, then it has been shown that, for the multicast problem (one source and many destinations), the capacity reduces by $2z$ compared to the non-Byzantine problem \cite{CaiYeung1:06,CaiYeung2:06}. To achieve this rate, only linear network coding is needed. Furthermore, if there is just one source and one destination, coding is needed only at the source node; internal nodes need only do routing.

The above model assumes that any set of $z$ edges may be adversarial, which may not accurately reflect all types of attacks. This model is accurate if the attacker is able to cut transmission lines and change messages that are sent along them. However, if instead the attacker is able to seize a router in a network, it will control the values on all links connected to that router. Depending on which router is attack, the number of the links controlled by the adversary may vary. In an effort to more accurately model this situation, in this paper we assume that the adversary may control any set of $s$ nodes.

Defeating node-based attacks is fundamentally different from defeating edge-based attacks. First, the edge problem does not immediately solve the node problem. Consider, for example, the Cockroach network, shown in Fig.~\ref{fig:cock}. Suppose we wish to handle any single adversarial node in the network (i.e. $s=1$). One simple approach would be to apply to edge result from \cite{CaiYeung1:06,CaiYeung2:06}: no node controls more than two unit-capacity edges, so we can defeat the node-based attack by using a code that can handle an attack on any two edges. However, note that the achievable rate for this network without an adversary is 4, so subtracting twice the number of bad edges leaves us with an achievable rate of 0. As we will show, the actual capacity of the Cockroach network with one traitor node is 2. In effect, relaxing the node attack problem to the edge attack problem is too pessimistic, and we can do better if we treat the node problem differently.

Node-based attacks and edge-based attacks differ in an even more fundamental way. When the adversary may sieze control of any set of $z$ unit-capacity edges, it is clear that it should always take over edges on the minimum cut of the network. However, if the adversary may sieze any $s$ nodes, it is not so obvious: it may face a choice between a node directly on the min-cut, but with few output edges, and a node away from the min-cut, but with many output edges. For example, in the Cockroach network, node 4 has only one output edge, but it is on the min-cut (which is between nodes $S,1,2,3,4,5$ and $D$); node 1 has two output edges, so apparently more power, but it is one step removed from the min-cut, and therefore its power may be diminished. This uncertainty about where a network is most vulnerable seems to make the problem hard. Indeed, we find that many standard network coding techniques fail to achieve capacity, so we resort to nonlinear codes, and in particular Polytope Codes, to be described.

\begin{figure}
\centerline{\begin{psfrags}\footnotesize
\psfrag{s}[c]{$S$}
\psfrag{d}[c]{$D$}
\psfrag{n1}[c]{$1$}
\psfrag{n2}[c]{$2$}
\psfrag{n3}[c]{$3$}
\psfrag{n4}[c]{$4$}
\psfrag{n5}[c]{$5$}
\psfrag{x1d}[c]{}
\psfrag{x14}[c]{}
\psfrag{x24}[c]{}
\psfrag{x25}[c]{}
\psfrag{x35}[c]{}
\psfrag{x3d}[c]{}
\psfrag{x4d}[c]{}
\psfrag{x5d}[c]{}
\includegraphics[scale=.6]{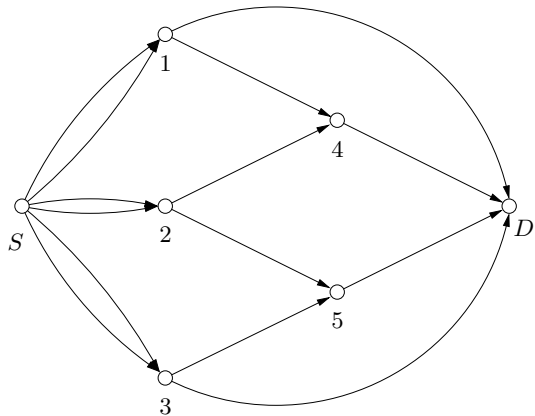}
\end{psfrags}}
\caption{The Cockroach Network. All edges have capacity 1. The capacity is 2, but no linear code can achieve a rate higher than 4/3. A proof of the linear capacity is given in Appendix~\ref{appendix:linear}. A capacity-achieving linear code supplemented by nonlinear comparisons is given in Sec.~\ref{sec:linearplus}, and a capacity-achieving Polytope Code is given in Sec.~\ref{sec:cockroach}.}
\label{fig:cock}
\end{figure}

\subsection{Related Work}

Byzantine attacks on network coding were first studied in \cite{HoEtal:04ISIT}, which looked at detecting adversaries in a random linear coding environment. The $z$ unit-capacity edge adversary problem was solved in \cite{CaiYeung1:06,CaiYeung2:06}. In \cite{JaggiEtal:07INFOCOM}, the same problem is studied, providing distributed and low complexity coding algorithms to achieve the same asymptotically optimal rates. In addition, \cite{JaggiEtal:07INFOCOM} looks at two adversary models slightly different from the omniscient one considered in \cite{CaiYeung1:06,CaiYeung2:06} and in this paper. They show that higher rates can be achieved under these alternate models. In \cite{KoetterKschishang:08IT}, a more general view of the adversary problem is given, whereby the network itself is abstracted into an arbitrary linear transformation.

Network coding under Byzantine attacks that are more general than the simple edge-based model was first studied in \cite{KosutTongTse:09Allerton}, a conference version of this work, and \cite{KimEtal:09Allerton}. The latter looked at the problem of edge-based attacks when the edges may have unequal capacities. This problem was found to have similar complications to the node-based problem. In particular, both \cite{KosutTongTse:09Allerton} and \cite{KimEtal:09Allerton} found that linear coding is suboptimal, and that simple nonlinear operations used to augment a linear code can improve throughput. Indeed, \cite{KimEtal:09Allerton} used a network almost identical to what we call the Cockroach network to demonstrate that nonlinear operations are necessary for the unequal edge problem. We show in Sec.~\ref{sec:loose} that the unequal-capacity edge problem is subsumed by the node problem.

These works seek to correct for the adversarial errors at the destination. An alternative strategy known as the watchdog, studied for wireless network coding in \cite{KimEtal:09ISIT}, is for nodes to police downstream nodes by overhearing their messages to detect modifications. In \cite{LiangAgarwalVaidya:10INFOCOM}, a similar approach is taken, and they found that nonlinear operations similar to ours can be helpful, in which comparisons are made to detect errors.

\subsection{Main Results}

Many achievability results in network coding have been proved using linear codes over a finite field. In this paper we demonstrate that linear codes are insufficient for this problem. Moreover, we develop a class of codes called Polytope Codes, originally introduced in \cite{KosutTongTse:09Allerton} under the less descriptive term ``bounded-linear codes''. Polytope codes are used to prove that a cut-set bound, stated and proved in Sec.~\ref{sec:cutset}, is tight for a certain class of networks. Polytope Codes differ from linear codes in three ways:
\begin{enumerate}
\item \emph{Comparisons:} A significant tool we use to defeat the adversary is that internal nodes in the network perform comparisons: they check whether their received data could have occurred if all nodes had been honest. If not, then there must be an upstream traitor that altered one of the received values, in which case this traitor can be localized. The result of the comparison, a bit signifying whether or not it succeeded, can be transmitted downstream through the network. The destination receives these comparison bits and uses them to determine who may be the traitors, and how to decode. These comparison operations are nonlinear, and, as we will demonstrate in Sec.~\ref{sec:linearplus}, incorporating them into a standard finite-field linear code can increase achieved rate. However, even standard linear codes supplemented by these nonlinear comparison operations appears to be insufficient to achieve capacity in general. Polytope Codes also incorporate comparisons, but of a more sophisticated variety.
\item \emph{Constant Composition Codebooks:} Unlike usual linear network codes, Polytope Codes are essentially constant composition codes. In particular, each Polytope Code is governed by a joint probability distribution on a set of random variables, one for each edge in the network. The codebook is composed of the set of all sequences with joint type exactly equal to this distribution. The advantage of this method of code construction is that an internal node knows exactly what joint type to expect of its received sequences, because it knows the original distribution. In a Polytope Code, comparisons performed inside the network consist of checking whether the observed joint type matches the expected distribution. If it does not, then the adversary must have influenced one of the received sequences, so it can be localized.
\item \emph{Distributions over Polytopes:} The final difference between linear codes and Polytope Codes---and the one for which the latter are named---comes from the nature of the probability distributions that, as described above, form the basis of the code. They are uniform distributions over the set of integer lattice points on polytopes in real vector fields. This choice for distribution provides two useful properties. First, the entropy vector for these distributions can be easily calculated merely from properties of the linear space in which the polytope sits. For this reason, they share characteristics with finite-field linear codes. In fact, a linear code can almost always be converted into a Polytope Code achieving the same rate. (There would be no reason to do this in practice, since Polytope Codes require much longer blocklengths.) The second useful property has to do with how the comparisons inside the network are used. These distributions are such that if enough comparisons succeed, the adversary is forced to act as an honest node and transmit correct information. We consider this to be the fundamental property of Polytope Codes. It will be elaborated in examples in Sec.~\ref{sec:caterpillar} and Sec.~\ref{sec:cockroach}, and then stated in its most general form as Theorem~\ref{thm:magic} in Sec.~\ref{sec:polytope}.
\end{enumerate}

We state in Sec.~\ref{sec:planarthm} our result that the cut-set bound can be achieved using Polytope Codes for a class of planar networks. Planarity requires that the graph can be embedded in a plane such that intersections between edges occur only at nodes. This ensures that enough opportunities for comparisons are available, allowing the code to more well defeat the adversary. The theorem is proved in Sec.~\ref{sec:planarpf}, but first we develop the theory of Polytope Codes through several examples in Sec.~\ref{sec:linearplus}--\ref{sec:cockroach}. In addition, we show in Sec.~\ref{sec:loose} that the cut-set bound is not always tight, by giving an example with a tighter bound. We conclude in Section~\ref{sec:conclusion}.

\section{Problem Formulation}\label{sec:form}

Let $(V,E)$ be an directed acyclic graph. We assume all edges are unit-capacity, and there may be more than one edge connecting the same pair of nodes. One node in $V$ is denoted $S$, the source, and one is denoted $D$, the destination. We wish to determine the maximum achievable communication rate from $S$ to $D$ when any set of $s$ nodes in $V\setminus\{S,D\}$ are \emph{traitors}; i.e. they are controlled by the adversary. Given a rate $R$ and a block-length $n$, the message $W$ is chosen at random from the set $\{1,\ldots,2^{nR}\}$. Each edge $e$ holds a value $X_e\in \{1,\ldots,2^{n}\}$.

A code is be made up of three components:
\begin{enumerate}
\item an encoding function at the source, which takes the message as input and produces values to place on all output edges,
\item a coding function at each internal node $i\in V\setminus\{S,D\}$, which takes the values on all input edges to $i$, and produces values to place on all output edges from $i$,
\item and a decoding function at the destination, which takes the values on all input edges and produces an estimate $\hat{W}$ of the message.
\end{enumerate}

Suppose $T\subseteq V\setminus\{S,D\}$ is the set of traitors, with $|T|=s$. They may subvert the coding functions at nodes $i\in T$ by placing arbitrary values on all the output edges from these nodes. Let $Z_T$ be the set of values on these edges. For a particular code, specifying the message $W$ as well as $Z_T$ determines exactly the values on all edges in the network, in addition to the destination's estimate $\hat{W}$. We say that a rate $R$ is \emph{achievable} if there exists a code operating at that rate with some block-length $n$ such that for all messages, all sets of traitors $T$, and all values of $Z_T$, $W=\hat{W}$. That is, the destination always decodes correctly no matter what the adversary does. Let the \emph{capacity} $C$ be the supremum over all achievable rates.

\section{Cut-Set Upper Bound}\label{sec:cutset}

It is shown in \cite{CaiYeung1:06,CaiYeung2:06} that, if an adversary controls $z$ unit-capacity edges, the network coding capacity reduces by $2z$. This is a special case of a more general principle: an adversary-controlled part of the network does twice as much damage in rate as it would if that part of the network were merely removed. This doubling effect is for the same reason that, in a classical error correction code, the Hamming distance between codewords must be at least twice the number of errors to be corrected; this is the Singleton bound \cite{Singleton:IT64}. We now give a cut-set upper bound for node-based adversaries in network coding that makes this explicit.

A \emph{cut} in a network is a subset of nodes $A\subset V$ containing the source but not the destination. The cut-set upper bound on network coding without adversaries is the sum of the capacities of all forward-facing edges; that is, edges $(i,j)$ with $i\in A$ and $j\notin A$. All backward edges are ignored. In the adversarial problem, backward edges are more of a concern. This is because the bound relies on messages that are sent along edges not controlled by the adversary being unaffected by those that are, which is not guaranteed in the presence of a backwards edge. We give an example of this in Appendix~\ref{appendix:cutset}. To avoid the complication, we state here a simplified cut-set bound that applies only to cuts without backward edges. This bound will be enough to find the capacity of the class of planar networks to be specified in Sec.~\ref{sec:planarthm}, but for the general problem it can be tightened. We state and prove a tighter version of the cut-set bound in Appendix~\ref{appendix:cutset}. Unlike the problem without adversaries, we see that there is no canonical notion of a cut-set bound. Some even more elaborate bounds are found in \cite{KimEtal:09Allerton,KimEtal:10ITA}. These papers study the unequal-edge problm, but the bounds can be readily applied to the node problem.

It was originally conjectured in \cite{KimEtal:09Allerton} that even the best cut-set bound is not tight in general. In Sec.~\ref{sec:loose}, we demonstrate that there can be an active upper bound on capacity fundamentally unlike a cut-set bound. The example used to demonstrate this, though it is a node adversary problem, can be easily modified to confirm the conjecture stated in \cite{KimEtal:09Allerton}.

\begin{theorem}\label{thm:cutset}
Consider a cut $A\subset V$ with $S\in A$ and $D\notin A$ and with no backward edges; that is, there is no edge $(i,j)\in E$ with $i\notin A$ and $j\in A$. If there are $s$ traitor nodes, then for any set $U\subset V\setminus\{S,D\}$ with $|U|=2s$, the following upper bound holds on the capacity of the network:
\beq C\le|\{(i,j)\in E:i\in A\setminus U,j\notin A\}|.\label{eq:cutset}\eeq
\end{theorem}

\begin{figure}
\centerline{
\begin{psfrags}\footnotesize
\psfrag{S}[c]{$S$}
\psfrag{D}[c]{$D$}
\psfrag{A}[c]{$A$}
\psfrag{Ac}[c]{$A^c$}
\psfrag{E1}[c]{$E_1$}
\psfrag{E2}[c]{$E_2$}
\psfrag{T1}[c]{$T_1$}
\psfrag{T2}[c]{$T_2$}
\psfrag{Ebar}[c]{$\bar{E}$}
\psfrag{X1}[c]{$X_{E_1}(w_2)$}
\psfrag{X2}[c]{$X_{E_2}(w_1)$}
\psfrag{Xbar}[c]{$X_{\bar{E}}(w_1)=X_{\bar{E}}(w_2)$}
\includegraphics[scale=.6]{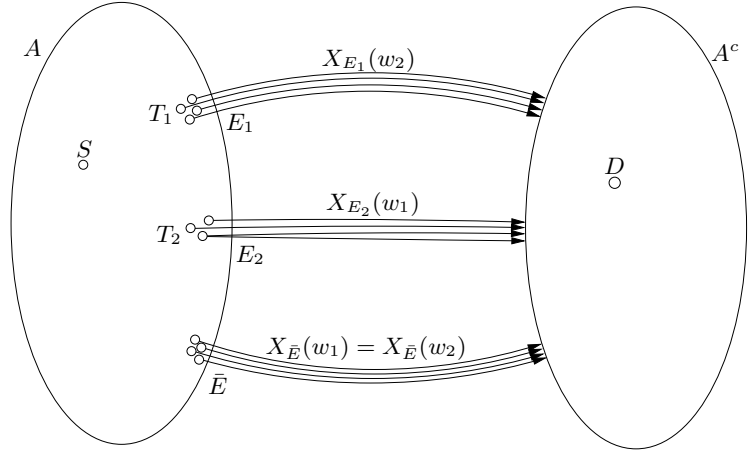}
\end{psfrags}}
\caption{Diagram of the proof of Theorem~\ref{thm:cutset}. The values on the links crossing the cut are such that it is impossible to determine whether $T_1$ or $T_2$ is the true set of traitors, and which of $w_1$ or $w_2$ is the true message.}
\label{fig:cutfig}
\end{figure}

\begin{proof}
Divide $U$ into two sets $T_1$ and $T_2$ with $|T_1|=|T_2|=s$. Let $E_1$ and $E_2$ be the sets of edges out of nodes in $T_1$ and $T_2$ respectively that cross the cut; that is, edges $(i,j)$ with $i\in A$ and $j\notin A$. Let $\bar{E}$ be the set of all edges crossing the cut not out of nodes in $T_1$ or $T_2$. Observe that the upper bound in \eqref{eq:cutset} is precisely the total capacity for all edges in $\bar{E}$. Note also that, since there are no backwards edges for the cut $A$, the values on edges in $\bar{E}$ are not influenced by the values on edges in $E_1$ or $E_2$. This setup is diagrammed in Fig.~\ref{fig:cutfig}.

Suppose \eqref{eq:cutset} does not hold. Therefore there exists a code with block-length $n$ achieving a rate $R$ higher than the right hand side of \eqref{eq:cutset}. For any set of edges $F\subseteq E$, for this code, we can define a function
\beq X_F:2^{nR}\to \prod_{e\in F}2^{n}\eeq
such that for a message $w$, assuming all nodes act honestly, the values on edges in $F$ is given by $X_F(w)$. Since $R$ is greater than the total capacity for all edges in $\bar{E}$, there exists two messages $w_1$ and $w_2$ such that $X_{\bar{E}}(w_1)=X_{\bar{E}}(w_2)$.

We demonstrate that it is possible for the adversary to confuse the message $w_1$ with $w_2$. Suppose $w_1$ were the true message, and the traitors are $T_1$. The traitors replace the values sent along edges in $E_1$ with $X_{E_1}(w_2)$. If there are edges out of nodes in $T_1$ that are not in $E_1$---i.e. they do not cross the cut---the traitors do not alter the values on these edges. Thus, the values sent along edges in $\bar{E}$ are given by $X_{\bar{E}}(w_1)$. Now suppose $w_2$ were the true message, and the traitors are $T_2$. They replace the messages going along edges in $E_2$ with $X_{E_2}(w_1)$, again leaving all other edges alone. Note that in both these cases, the values on $E_1$ are $X_{E_1}(w_2)$, the values on $E_2$ are $X_{E_2}(w_1)$, and the values on $\bar{E}$ are $X_{\bar{E}}(w_1)$. This comprises all edges crossing the cut, so the destination receives the same values under each case; therefore it cannot distinguish $w_1$ from $w_2$.
\end{proof}

We illustrate the use of Theorem~\ref{thm:cutset} on the Cockroach network, as shown in Fig.~\ref{fig:cock}, with a single adversary node. To apply the bound, we choose a cut $A$ and a set $U$ with $|U|=2s=2$. Take $A=\{S,1,2,3,4,5\}$, and $U=\{1,4\}$. Four edges cross the cut, but the only ones not out of nodes $U$ are $(3,D)$ and $(5,D)$, so we may apply Theorem~\ref{thm:cutset} to give an upper bound on capacity of 2. Alternatively, we could take $A=\{S,1,2,3\}$ and $U=\{1,2\}$, to give again an upper bound of 2. Note that there are 6 edges crossing this second cut, even though the cut-set bound is the same. It is not hard to see that 2 is the smallest upper bound given by Theorem~\ref{thm:cutset} for the capacity of the Cockroach network. In fact, rate 2 is achievable, as will be shown in Sec.~\ref{sec:linearplus} using a linear code supplemented by comparison operations, and again in Sec.~\ref{sec:cockroach} using a Polytope Code.

\section{Capacity of A Class of Planar Networks}\label{sec:planarthm}

\begin{theorem}\label{thm:planar}
Let $(V,E)$ be a network with the following properties:
\begin{enumerate}
\item It is planar.
\item No node other than the source has mare than two unit-capacity output edges.
\item No node other than the source has more output edges than input edges.
\item There is at most one traitor (i.e. $s=1$).
\end{enumerate}
The cut-set bound given by Theorem~\ref{thm:cutset} is tight for $(V,E)$.
\end{theorem}

Polytope Codes are used to prove achievability for this theorem. The complete proof is given in Sec.~\ref{sec:planarpf}, but first we develop the theory of Polytope Codes by means of several examples in Sec.~\ref{sec:linearplus}--\ref{sec:cockroach} and general properties in Sec.~\ref{sec:polytope}.

Perhaps the most interesting condition in the statement of Theorem~\ref{thm:planar} is the planarity condition. Recall that a graph is said to be \emph{embedded} in a surface (generally a two dimensional manifold) when it is drawn in this surface so that edges intersect only at nodes. A graph is \emph{planar} if it can be embedded in the plane.

\section{A Linear Code with Comparisons for the Cockroach Network}\label{sec:linearplus}

The Cockroach network satisfies the conditions of Theorem~\ref{thm:planar}. Fig~\ref{fig:cock} shows a plane embedding with both $S$ and $D$ on the exterior, and the second and third conditions are easily seen to be satisfied. Therefore, since the smallest cut-set bound given by Theorem~\ref{thm:cutset} for a single traitor node is 2, Theorem~\ref{thm:planar} claims that the capacity of the Cockroach network is 2. In this section, we present a capacity-achieving code for the Cockroach network that is composed of a linear code over a finite-field supplemented by nonlinear comparisons. This illustrates the usefulness of comparisons in defeating Byzantine attacks on network coding. Before doing so, we provide an intuitive argument that linear codes are insufficient. A more technical proof that the linear capacity is in fact 4/3 is given in Appendix~\ref{appendix:linear}.

Is it possible to construct a linear code achieving rate 2 for the Cockroach network? We know from the Singleton bound-type argument---the argument at the heart of the proof of Theorem~\ref{thm:cutset}---that, in order to defeat a single traitor node, if we take out everything controlled by two nodes, the destination must be able to decode from whatever remains. Suppose we take out nodes 2 and 3. These nodes certainly control the values on $(5,D)$ and $(3,D)$, so if we hope to achieve rate 2, the values on $(1,D)$ and $(4,D)$ must be uncorruptable by nodes 2 and 3. Edge $(1,D)$ is not a problem, but consider $(4,D)$. With a linear code, the value on this edge is a linear combination of the values on $(1,4)$ and $(2,4)$. In order to keep the value on $(4,D)$ uncorruptable by node 2, the coefficient used to construct the value on $(4,D)$ from $(2,4)$ must be zero. In other words, node 4 must ignore the value on $(2,4)$ when constructing the value it sends on $(4,D)$. If this is the case, we lose nothing by removing $(2,4)$ from the network. However, without this edge, we may apply Theorem~\ref{thm:cutset} with $A=\{S,1,2,3\}$ and $U=\{1,3\}$ to conclude that the capacity is no more than 1. Therefore no linear code can successfully achieve rate 2.

This argument does not rigorously show that the linear capacity is less than 2, because it shows only that a linear code cannot achieve exactly rate 2, but it does not bound the achievable rate with a linear code away from 2. However, it is meant to be an intuitive explanation for the limitations of linear codes for this problem, as compared with the successful nonlinear codes that we will subsequently present. The complete proof that the linear capacity is 4/3 is given in Appendix~\ref{appendix:linear}.

We now introduce a nonlinear code to achieve the capacity of 2. We work in the finite field of $p$ elements. Let the message $w$ be a $2k$-length vector split into two $k$-length vectors $x$ and $y$. We will use a block length large enough to place one of $2p^k$ values on each link. In particular, this is enough to place on a link some linear combination of $x$ and $y$, as well as one additional bit. For large enough $k$, this extra bit becomes insignificant, so we still achieve rate 2.

\begin{figure}\centerline{
\begin{psfrags}\footnotesize
\psfrag{s}[c]{$S$}
\psfrag{d}[c]{$D$}
\psfrag{n1}[c]{$1$}
\psfrag{n2}[c]{$2$}
\psfrag{n3}[c]{$3$}
\psfrag{n4}[c]{$4$}
\psfrag{n5}[c]{$5$}
\psfrag{x1d}[c]{$x$}
\psfrag{x14}[c]{$y$}
\psfrag{x24}[c]{$y$}
\psfrag{x25}[c]{$x+y$}
\psfrag{x35}[c]{$x+y$}
\psfrag{x3d}[c]{$x-y$}
\psfrag{x4d}[c]{$y$}
\psfrag{x5d}[c]{$x+y$}
\psfrag{x4d2}[c]{\tiny$(=,\ne)$}
\psfrag{x5d2}[c]{\tiny$(=,\ne)$}
\includegraphics[scale=.6]{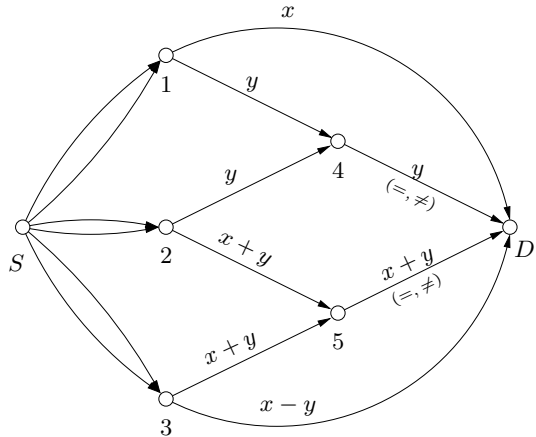}
\end{psfrags}}
\caption{A nonlinear code for the Cockroach Network achieving the capacity of 2.}
\label{fig:cocknonlinear}
\end{figure}

The scheme is shown in Figure~\ref{fig:cocknonlinear}. Node $4$ receives the vector $y$ from both nodes $1$ and $2$. It forwards one of these copies to $D$ (it does not matter which). In addition, it performs a nonlinear comparison between the two received copies of $y$, resulting in a bit comprised of one of the special symbols $=$ or $\ne$. If the two received copies of $y$ agree, it sends $=$, otherwise it sends $\ne$. The link $(4,D)$ can accommodate this, since it may have up to $2p^k$ messages placed on it. Node $5$ does the same with its two copies of the vector $x+y$.

The destination's decoding strategy depends on the two comparison bits sent from nodes $4$ and $5$, as follows:

\bi
\item If node $5$ sends $\ne$ but node $4$ sends $=$, then the traitor must be one of nodes $1$, $2$, or $4$. In any case, the vector $x-y$ received from node $3$ is certainly trustworthy. Moreover, $x+y$ can be trusted, because even if node $2$ is the traitor, its transmission must have matched whatever was sent by node $3$; if not, node $5$ would have transmitted $\ne$. Therefore the destination can trust both $x+y$ and $x-y$, from which it can decode the message $w=(x,y)$.

\item If node $5$ sends $\ne$ but node $4$ sends $=$, then we are in the symmetric situation and can reliably decode $w$ from $x$ and $y$.

\item If both nodes $4$ and $5$ send $\ne$, then the traitor must be node $2$, in which case the destination can reliable decode from $x$ and $x-y$.

\item If both messages are $=$, then the destination cannot eliminate any node as a possible traitor. However, we claim that at most one of $x,y,x+y,x-y$ can have been corrupted by the traitor. If node $1$ is the traitor, it may choose whatever it wants for $x$, and the destination would never know. However, node $1$ cannot impact the value of $y$ without inducing a $\ne$, because its transmission to node $4$ is verified against that from node $2$. Similarly, node $3$ controls $x-y$ but not $x+y$. Nodes $4$ and $5$ control only $y$ and $x+y$ respectively. Node $2$ controls nothing, because both $y$ and $x+y$ are checked against other transmissions. Therefore, if the destination can find three of $x,y,x+y,x-y$ that all agree on the message $w$, then this message must be the truth because only one of the four could be corrupted, and $w$ can be decoded from the other two. Conversely, there must be a group of three of $x,y,x+y,x+2y$ that agree, because at most one has been corrupted. Hence, the destination can always decode $w$.
\ei

Even though our general proof of Theorem~\ref{thm:planar} uses a Polytope Code, which differs significantly from this one, the manner in which the comparisons comes into play is essentially the same. The key insight is to consider the code from the perspective of the traitor. Suppose it is node 1, and consider the choice of what value for $y$ to send along edge $(1,4)$. If it sends a false value for $y$, then the comparison at node 4 will fail, which will lead the destination to consider the upper part of the network suspect, and thereby ignore all values influenced by node 1. The only other choice for node 1 is to cause the comparison at node 4 to succeed; but this requires sending the true value of $y$, which means it has no hope to corrupt the decoding process. This is the general principle that makes our codes work: force the to make a choice between acting like an honest node, or acting otherwise and thereby giving away its position.

We make one further note on this code, having to do with why the specific approach used here for the Cockroach network fails on the more general problem. Observe that in order to make an effective comparison, the values sent along edges $(1,4)$ and $(2,4)$ needed to be exactly the same. If they had been independent vectors, no comparison could be useful. This highly constrains the construction of the code, and even though it succeeds for this network, it fails for others, such as the Caterpillar network, to be introduced in the next section. The advantage of the Polytope Code is that it deconstrains the types of values that must be available in order to form a useful comparison; in fact, it becomes possible to have useful comparisons between nearly independent variables, which is not possible with a code built on a finite-field.

\section{An Example Polytope Code: The Caterpillar Network}\label{sec:caterpillar}

The Caterpillar Network is shown in Figure~\ref{fig:caterpillar}. We consider a slightly different version of the node-based attack on this network: at most one node may be a traitor, but only nodes 1--4. This network is not in the class defined in the statement of Theorem~\ref{thm:planar}, but we introduce it in order to motivate the Polytope Code.

\begin{figure}
\centerline{
\begin{psfrags}
\psfrag{s}[c]{\small $S$}
\psfrag{d}[c]{\small $D$}\footnotesize
\psfrag{x1}[c]{$$}
\psfrag{x2}[c]{$$}
\psfrag{x3}[c]{$$}
\psfrag{x4}[c]{$$}
\psfrag{y1}[c]{}
\psfrag{y2}[c]{}
\includegraphics[scale=.6]{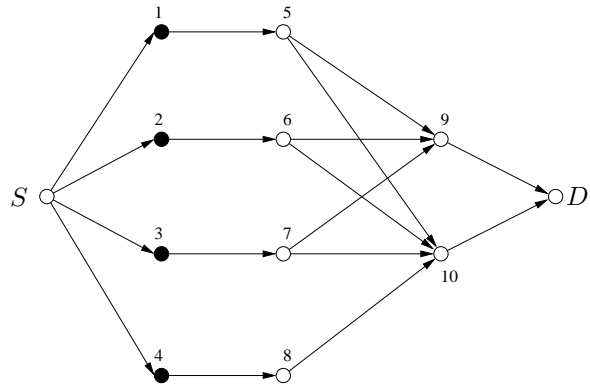}
\end{psfrags}}
\caption{The Caterpillar Network. All edges have unit capacity. One node may be a traitor, but only one of the black nodes: nodes 1--4.}
\label{fig:caterpillar}
\end{figure}

Even though this problem differs from the one defined earlier in that not every node in the network may be a traitor, it is easy to see that we may still apply the cut-set bound of Theorem~\ref{thm:cutset} as long as we take the set $U$ to be a subset of the allowable traitors. If we apply Theorem~\ref{thm:cutset} with $A=\{S,1,2,3,4\}$ and $U=\{1,2\}$, we find that the capacity of this network is no more than 2. As we will show, the capacity is 2.

Consider what is required to achievae rate 2. Of the four values on the edges $(1,5)$, $(2,6)$, $(3,7)$, and $(4,8)$, one may be corrupted by the adversary. This means that these four values must form a $(4,2)$ MDS code. That is, given any uncorrupted pair of these four values, it must be possible to decode the message exactly. Since each edge has capacity 1, in order to achieve rate 2, the values on each pair of edges must be independent, or nearly independent. For example, we could take the message to be composed of two elements $x,y$ from a finite field, and transmit on these four edges $x,y,x+y,x-y$. However, as we will show, this choice does not succeed.

Now consider the two edges $(9,D)$ and $(10,D)$. As these are the only edges incident to the destination, to achieve rate 2, both must hold values guaranteed to be uncorrupted by the traitor. We may assume that nodes 5--8 forward whatever they receive on their incoming edges to all their outgoing edges, so node 10 receives all four values sent from nodes 1--4. From these, it can decode the entire message, so it is not a problem for it to construct a trustworthy value to send along $(10,D)$. However, node 9 has access to only three of the four values sent from nodes 1--4, from which it is not clearly possible to construct a trustworthy value. The key problem, then, is to design the values on edges $(1,5),(2,6),(3,7)$ to be pairwise independent, but such that if one value is corrupted, it is always possible to construct a trustworthy value to transmit on $(9,D)$. This is impossible to do using a finite field code. For example, suppose node 9 is given values for $x,y,x+y$, one of which may be corrupted by the traitor. If the linear constraint among these three values does not hold---that is, if the received value for $x+y$ does not match the sum of the value for $x$ and the value for $y$---then any of the three values may be the incorrect one. Therefore, from node 9's perspective, any of nodes 1, 2, or 3 could be the traitor. In order to produce a trustworthy symbol, it must be able to correctly conclude that one of these three nodes is not the traitor. If, for example, it could determine that the traitor was node 1 or node 2 but not node 3, then the value sent along $(3,7)$ could be forwarded to $(9,D)$ with a guarantee of correctness. A linear code over a finite field does not allow this, but a Polytope Code does.

\subsection{Coding Strategy}

We now begin to describe a capacity-achieving Polytope Code for the Caterpillar network. We do so first by describing how the code is built out of a probability distribution, and the properties we need this probability distribution to satisfy. Subsequently, we give an explicit construction for a probability distribution derived from a polytope in a real vector field, and show that it has the desired properties.

Let $X,Y,Z,W$ be jointly distributed random variables, each defined over the finite alphabet $\Xmsc$. Assume all probabilities on these random variables are rational. Let $T^n(XYZW)\in\Xmsc^{4n}$ be the set of joint sequences $(x^ny^nz^nw^n)$ with joint type exactly equal to the distribution on $X,Y,Z,W$. For $n$ such that $T^n(XYZW)$ is not empty, we know from the theory of types that
\beq |T^n(XYZW)|\ge\frac{1}{(n+1)^{|\Xmsc|^4}}2^{nH(XYZW)}.\eeq
Our coding strategy will be to associate each element of $T^n(XYZW)$ with a distinct message. Given the message, we find the associated four sequences $x^n,y^n,z^n,w^n$, and transmit them on the four edges out of nodes 1,2,3,4 respectively. Doing this requires placing a sequence in $\Xmsc^n$ on each edge. The rate of this code is
\beq \frac{\log|T^n(XYZW)|}{n\log|\Xmsc|}\ge \frac{H(XYZW)}{\log|\Xmsc|}-\frac{|\Xmsc|^4\log(n+1)}{n\log|\Xmsc|}.\label{eq:ratefromtype}\eeq
Note that for sufficiently large $n$, we may operate at a rate arbitrarily close to $\frac{H(XYZW)}{\log|\Xmsc|}$. Therefore, to achieve rate 2, we would like the following property to hold.
\begin{property}\label{prop:rate}
 $\frac{H(XYZW)}{\log|\Xmsc|}=2$.
\end{property}

The adversary may alter the value of one of the sequences sent out of nodes 1--4. By the Singleton bound argument, it must be possible to reconstruct the message from any two of these four sequences. We therefore need the following property.
\begin{property}\label{prop:depend}
 Any two of $X,Y,Z,W$ determine the other two.
\end{property}

For reasons that will become clear, we also need the following property. It is an example of the fundamental property of Polytope Codes.
\begin{property}\label{prop:pairmagic} Any random variables $\tilde{X},\tilde{Y},\tilde{Z}$ satisfying the three conditions
\begin{align}
 (\tilde{X},\tilde{Y})&\sim(X,Y)\label{eq:pq1}\\
(\tilde{X},\tilde{Z})&\sim(X,Z)\label{eq:pq2}\\
(\tilde{Y},\tilde{Z})&\sim(Y,Z)\label{eq:pq3}
\end{align}
also satisfy
\beq (\tilde{X},\tilde{Y},\tilde{Z})\sim(X,Y,Z).\label{eq:pq4}\eeq
\end{property}

Suppose we are given random variables $X,Y,Z,W$ satisfying Properties~\ref{prop:rate}--\ref{prop:pairmagic}. We now describe what nodes 9 and 10 transmit to the destination. Let $\tilde{x}^n,\tilde{y}^n,\tilde{z}^n,\tilde{w}^n$ be the four sequences that are sent on the edges out of nodes 1--4; because of the traitor at most one of these may differ from $x^n,y^n,z^n,w^n$. Let random variables $\tilde{X},\tilde{Y},\tilde{Z},\tilde{W}$ have joint distribution equal to the joint type of $(\tilde{x}^n,\tilde{y}^n,\tilde{z}^n,\tilde{w}^n)$. This is a formal definition; these variables do not exist per se in the network, but defining them make it convenient to describe the behavior of the code. Since node 9 recevies $\tilde{x}^n,\tilde{y}^n,\tilde{z}^n$, it knows exactly the joint distribution of $\tilde{X},\tilde{Y},\tilde{Z}$. In particular, it can check which of \eqref{eq:pq1}--\eqref{eq:pq4} are satisfied for these variables.

Suppose \eqref{eq:pq4} holds. Then all three sequences $\tilde{x}^n,\tilde{y}^n,\tilde{z}^n$ are trustworthy, because if a traitor is among nodes 1--3, it must have transmitted the true value of its output sequence, or else the empirical type would not match, due to Property~\ref{prop:depend}. In this case, node 9 forwards $\tilde{x}^n$ to the destination, confident that it is correct. Meanwhile, node 10 can also observe $\tilde{X},\tilde{Y},\tilde{Z}$, and so it forwards $\tilde{y}^n$ to the destination.

Now suppose \eqref{eq:pq4} does not hold. Then by Property~\ref{prop:pairmagic}, one of \eqref{eq:pq1}--\eqref{eq:pq3} must not hold. Suppose, for example, that $(\tilde{X},\tilde{Y})\not\sim(X,Y)$. Because of our constant composition code construction, this can only occur if either node 1 or 2 is the traitor. Hence, node 3 is honest, so Node 9 may forward $\tilde{z}^n$ to the destination without error. Similarly, no matter which pairwise distribution does not match, node 9 can always forward the sequence not involved in the mismatch. Meanwhile, node 10 may forward $\tilde{w}^n$ to the destination, since in any case the traitor has been localized to nodes 1--3. The destination always receives two of the four sequences, both guaranteed correct; therefore it may decode.

\subsection{The Polytope Distribution}

All that remains to prove that rate 2 can be achieved for the Caterpillar network is to show that there exists variables $X,Y,Z,W$ such that Properties~\ref{prop:rate}--\ref{prop:pairmagic} hold. In fact, this is not quite possible. In particular, Property~\ref{prop:rate} implies that $X,Y,Z,W$ are pairwise independent. If so, Property~\ref{prop:pairmagic} cannot hold, because we can take $\tilde{X},\tilde{Y},\tilde{Z}$ to be jointly independent with $\tilde{X}\sim X$, $\tilde{Y}\sim Y$, and $\tilde{Z}\sim Z$. This satisfies \eqref{eq:pq1}--\eqref{eq:pq3} but not \eqref{eq:pq4}. We therefore replace Property~\ref{prop:rate} with the following slight relaxation.
\begin{property}\label{prop:ratealt}
 $\frac{H(XYZW)}{\log|\Xmsc|}\ge 2-\eps$.
\end{property}
If for every $\eps>0$, there exists a set of random variables satisfying Properties~\ref{prop:depend}--\ref{prop:ratealt}, then by \eqref{eq:ratefromtype} we achieve rate 2.

The most unusual aspect of the Polytope Code is Property~\ref{prop:pairmagic} and its generalization, to be stated as Theorem~\ref{thm:magic} in Sec.~\ref{sec:polytope}. Therefore, before constructing a distribution satisfying all three properties, we illustrate in Table~\ref{table:simpdist} a very simple distribution on three binary variables variables that satisfy just Property~\ref{prop:pairmagic}. This distribution is only on $X,Y,Z$; for simplicity leave out $W$, as it is not involved in Property~\ref{prop:pairmagic}. We encourage the reader to manually verify Property~\ref{prop:pairmagic} for this distribution. Observe that $X,Y,Z$ as given in Table~\ref{table:simpdist} may be alternatively expressed as being uniformly distributed on the following polytope:
\beq\big\{x,y,z\in\{0,1\}:x+y+z=1\big\}.\eeq
This is a special case of the construction of the distributions in the sequel.

\begin{table}\caption{A simple distribution satisfying Property~\ref{prop:pairmagic}.}\label{table:simpdist}
\[\begin{array}{cccc}
x&y&z&p(xyz)\\
\hline
0&0&0&0\\
0&0&1&1/3 \\
0&1&0&1/3\\
0&1&1&0\\
1&0&0&1/3\\
1&0&1&0\\
1&1&0&0\\
1&1&1&0
\end{array}\]
\end{table}

We now construct a distribution satisfying Properties~\ref{prop:depend}--\ref{prop:ratealt} for arbitrarily small $\eps$. Take $k$ to be a positive integer, and let $X,Y,Z,W$ be uniform over the set
\beq\big\{(x,y,z,w)\in\{-k,\ldots,k\}^4: x+y+z=0\text{ and }3x - y + 2w = 0\big\}.\label{eq:abcd}\eeq
Note that this is the set of integer lattice points in a polytope.

By the linear constraints in \eqref{eq:abcd}, this distribution satisfies Property~\ref{prop:depend}. Now consider Property~\ref{prop:ratealt}. The region of $(X,Y)$ pairs with positive probability is shown in Figure~\ref{fig:ab}. Note that $X$ and $Y$ are not independent, because the boundedness of $Z$ and $W$ requires that $X$ and $Y$ satisfy certain linear inequalities. Nevertheless, the area of the polygon shown in Figure~\ref{fig:ab} grows as $\Ocal(k^2)$. Hence
\beq \frac{\log H(XYZW)}{\log|\Xmsc|}=\frac{\log \Ocal(k^2)}{\log(2k+1)}\ge 2-\eps\eeq
where the last inequality holds for sufficiently large $k$. Thus these variables satisfy Property~\ref{prop:ratealt}.

\begin{figure}
\centerline{
\begin{psfrags}\scriptsize
\psfrag{x}[c]{$X$}
\psfrag{y}[c]{$Y$}
\psfrag{k}[c]{$k$}
\psfrag{mk}[c]{$-k$}
\centerline{\includegraphics[width=1.8in]{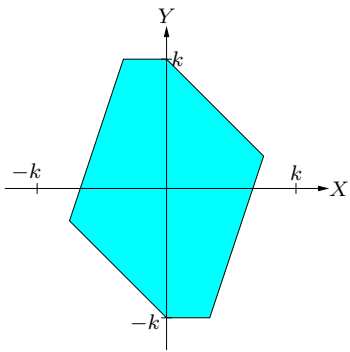}}
\end{psfrags}}
\caption{An example polytope projected into the $(x,y)$ plane.}
\label{fig:ab}
\end{figure}

We now consider Property~\ref{prop:pairmagic}. Assuming $\tilde{X},\tilde{Y},\tilde{Z}$ satisfy \eqref{eq:pq1}--\eqref{eq:pq3}, we may write
\begin{align}
\Embb\big[(\tilde{X}+\tilde{Y}+\tilde{Z})^2\big]
&=\Embb\big[\tilde{X}^2+\tilde{Y}^2+\tilde{Z}^2+2\tilde{X}\tilde{Y}+2\tilde{X}\tilde{Z}+2\tilde{Y}\tilde{Z}\big]\\
&=\Embb\big[X^2+Y^2+Z^2+2XY+2XZ+2YZ\big]\label{eq:xyzpf1}\\
&=\Embb\big[(X+Y+Z)^2]\\
&=0\label{eq:xyzpf2}\end{align}
where \eqref{eq:xyzpf1} holds from \eqref{eq:pq1}--\eqref{eq:pq3}, and because each term in the sum involves at most two of the three variables; and \eqref{eq:xyzpf2} holds because $X+Y+Z=0$ by construction. Hence $\tilde{X}+\tilde{Y}+\tilde{Z}=0$, so we may write
\begin{align}
(\tilde{X},\tilde{Y},\tilde{Z})
&=(\tilde{X},\tilde{Y},-\tilde{X}-\tilde{Y})\\
&\sim(X,Y,-X-Y)\label{eq:xyzpf3}\\
&=(X,Y,Z)
\end{align}
where \eqref{eq:xyzpf3} holds by \eqref{eq:pq1}. This verifies Property~\ref{prop:pairmagic}, and we may now conclude that the distribution on $X,Y,Z,W$ satisfies all desired properties, so the induced Polytope Code achieves rate 2 for the Caterpillar network.

The above argument took advantage of the linear constraint $X+Y+Z=0$, but this constraint was in no way special. Property~\ref{prop:pairmagic} would hold as long as $X,Y,Z$ are subject to any linear constraint with nonzero coefficients for all three variables.

Observe that when $k$ is large, any pair of the four variables are nearly independent, in that their joint entropy is close to the sum of their individual entropies. We have therefore constructed something like a $(4,2)$ MDS code. In fact, if we reinterpret the linear constraints in \eqref{eq:abcd} as constraints on elements $x,y,z,w$ from a finite field, the resulting finite subspace would be exactly a $(4,2)$ MDS code. This illustrates a general principle of Polytope Codes: any code construction on a finite field can be immediately used to construct a Polytope Code, and many of the properties of the original code will hold over. The resulting code will be substantially harder to implement, in that it involves much longer blocklengths, and more complicated coding functions, but additional properties, such as Property~\ref{prop:pairmagic}, may hold.

\section{A Polytope Code for the Cockroach Network}\label{sec:cockroach}

We return now to the Cockroach network, and demonstrate a capacity-achieving Polytope Code for it. We do this not to find the capacity for the network, because we have already done so with the simpler code in Sec.~\ref{sec:linearplus}, but rather to illustrate a Polytope Code on a network satisfying the conditions of Theorem~\ref{thm:planar}, which are of a somewhat different flavor than the Caterpillar network.

In Sec.~\ref{sec:linearplus}, we illustrated how performing comparisons and transmitting comparison bits through the network can help defeat traitors. In Sec.~\ref{sec:caterpillar}, we illustrated how a code can be built out a distribution on a polytope, and how a special property of that distribution comes into play in the operation of the code. To build a Polytope Code for the Cockroach network, we combine these two ideas: the primary data sent through the network comes from the distribution on a polytope, but then comparisons are performed in the network in order to localize the traitor.

The first step in constructing a Polytope Code is to describe a distribution over a polytope. That is, we define a linear subspace in a real vector field, and take a uniform distribution over the polytope defined by the set of vectors with entries in $\{-k,\ldots,k\}$ for some integer $k$. The nature of this distribution depends on the characteristics of the linear subspace. For our code for the Cockroach network, we need one that is equivalent to a $(6,2)$ MDS code. That is, the linear subspace sits in $\Rmbb^6$, has dimension 2, and is defined by four constraints such that any two variables determine the others. One choice for the subspace, for example, would be the set of $(a,b,c,d,e,f)$ satisfying
\begin{align}
a+b+c&=0\label{eq:62mds1}\\
a-b+d&=0\\
a+2b+e&=0\\
2a+b+f&=0.\label{eq:62mds4}
\end{align}
Let the random variables $A,B,C,D,E,F$ have joint distribution uniformly distributed over the polytope defined by \eqref{eq:62mds1}--\eqref{eq:62mds4} and $a,b,c,d,e,f\in\{-k,\ldots,k\}$. By a similar argument to that in Sec.~\ref{sec:caterpillar}, for large $k$,
\beq\frac{H(ABCDEF)}{\log(2k+1)}\approx 2.\eeq
We choose a block length $n$ and associate each message with a joint sequence $(a^nb^nc^nd^ne^nf^n)$ with joint type exactly equal to the distribution of the six variables. For large $n$ and $k$, we may place one sequence $a^n$--$f^n$ on each unit capacity edge in the network and operate near rate 2. These six sequences are generated at the source and then routed through the network as shown in Fig.~\ref{fig:cockcode}. For convenience, the figure omits the $n$ superscript, but we always mean them to be sequences.

\begin{figure}
\centerline{\begin{psfrags}\footnotesize
\psfrag{s}[c]{$S$}
\psfrag{d}[c]{$D$}
\psfrag{n1}[c]{$1$}
\psfrag{n2}[c]{$2$}
\psfrag{n3}[c]{$3$}
\psfrag{n4}[c]{$4$}
\psfrag{n5}[c]{$5$}
\psfrag{x1d}[c]{$a$}
\psfrag{x14}[c]{$b$}
\psfrag{x24}[c]{$c$}
\psfrag{x25}[c]{$d$}
\psfrag{x35}[c]{$e$}
\psfrag{x3d}[c]{$f$}
\psfrag{x4d}[c]{$c$}
\psfrag{x5d}[c]{$d$}
\psfrag{x4d2}[c]{\tiny$(=,\ne)$}
\psfrag{x5d2}[c]{\tiny$(=,\ne)$}
\includegraphics[scale=.6]{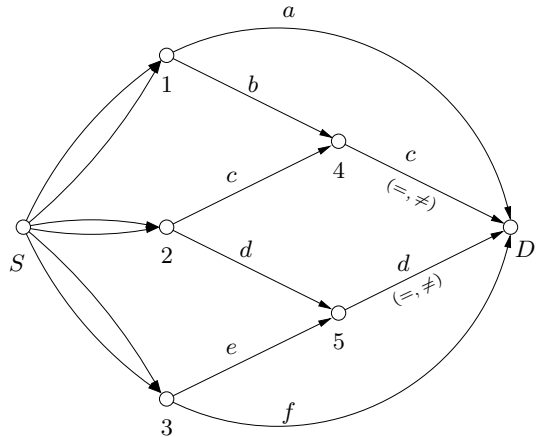}
\end{psfrags}}
\caption{A capacity-achieving Polytope Code for the Cockroach Network.}
\label{fig:cockcode}
\end{figure}

As in Sec.~\ref{sec:caterpillar}, we define $\tilde{A},\tilde{B},\tilde{C},\tilde{D},\tilde{E},\tilde{F}$ to have joint distribution equal to the type of the six sequences an they actually appear in the network, which may differ from the sequences sent by the source because of the adversary. In addition to forwarding one sequence, nodes 4 and 5 perform more elaborate operations. Like in the code for the Cockroach network described in Sec.~\ref{sec:linearplus}, they each perform a comparison and transmit either $=$ or $\ne$ depending on whether the comparison succeeds. In particular, they compare the types of their received sequences with the original distribution. For example, node 4 receives the two sequences $b^n$ and $c^n$, from which it can construct $\tilde{B}$ and $\tilde{C}$. If the joint distribution of $(\tilde{B},\tilde{C})$ matches that of $(B,C)$, it sends $=$ to the destination; if not, it sends $\ne$. This single bit costs asymptotically negligible rate, so it has no effect on the achieved rate of the code for large $n$ and $k$. Node 5 performs a similar action, comparing the distribution of $(\tilde{D},\tilde{E})$ with that of $(D,E)$, and transmitting a comparison bit to the destination.

We now describe the decoding operation at the destination. The first step is to compile a list of possible traitors. We denote this list $\Lmsc\subseteq\{1,\ldots,5\}$. The destination does this in the following way. Since the code is entirely known, with no randomness, the destination determines whether all its received data could be induced if each node were the traitor. That is, it considers each possible message, each possible traitor, and each possible set of values on the output edges of that traitor. Any combination of these determines the values received at the destination, which may be compared to what that the destination has in fact received. If a node $i$ is such that it could have been the traitor and induced the set of values received at the destination, for any message and any action by node $i$, then $i$ is put onto $\Lmsc$. This process ensures that the true traitor, even though it may not be known by the destination, is surely in $\Lmsc$. Note that this procedure could in principle be done for any code, not necessarily a Polytope Code.

Because we define $\Lmsc$ in this non-constructive manner, our arguments for code correctness may sometimes seem backwards. We will make assumptions about $\Lmsc$, and from there reason about the resulting constraints on what the traitor could have done, even though this is opposite to the causal relationship. We do this because it is most convenient to partition possible traitor actions based on the $\Lmsc$ that results. As long as our analysis considers every possible $\Lmsc$, we can be assured that the code can handle any possible traitor action.

Once $\Lmsc$ is determined, the next step in the decoding process is to use $\Lmsc$ to decide from which of the four symbols available at the destination to decode. Since any pair of the six original symbols contain all the information in the message, if at least two of the four symbols $a,c,d,f$ can be determined to be trustworthy by the destination, then it can decode. The destination discards any symbol that was touched by every node in $\Lmsc$, and decodes from the rest. For example, if $\Lmsc=\{2\}$, then the destination discards $c,d$ and decodes from $a,f$. If $\Lmsc=\{2,4\}$, the destination discards just $c$---because it is the only symbol touched by both nodes 2 and 4---and decodes from $a,d,f$. If $\Lmsc=\{1,\ldots,5\}$, then it discards no symbols and decodes from all four.

The prove the correctness of this code, we must show that the destination never decodes from a symbol that was altered by the traitor. This is easy to see if $|\Lmsc|=1$, because in this case the destination knows exactly which node is the traitor, and it simply discards all symbols that may have been influenced by this node. Since no node touches more than two of the symbols available at the destination, there are always at least two remaining from which to decode.

More complicated is when $|\Lmsc|\ge 2$. In this case, the decoding process, as described above, sometimes requires the destination to decode from symbols touched by the traitor. For example, suppose node 2 were the traitor, and $\Lmsc=\{1,2\}$. No symbols are touched by both nodes 1 and 2, so by the decoding rule the destination decodes using all four of its received symbols. In particular, the destination uses $c$ and $d$ to decode, even though both are touched by node 2. To prove correctness we must show that node 2 could not have transmitted anything but the true values of $c$ and $d$. What we use to prove this is the fact that $\Lmsc$ contains node 1, meaning that node 2 must have acted in a way such that it appears to the destination that node 1 could be the traitor. This induces constraints on the behavior of node 2. The first is that the comparison that occurs at node 5 between $d$ and $e$ must succeed. If it did not, then the destination would learn it, and conclude that node 1 could not be the traitor, in which case 1 would not be in $\Lmsc$. Hence the distribution of $(\tilde{D},\tilde{E})$ must match that of $(D,E)$. This constitutes a constraint on node 2 in its transmission of $d$. Moreover, $(\tilde{D},\tilde{F})\sim(D,F)$, because the destination may observe $d$ and $f$, so it could detect a difference between these two distributions if it existed. Because both symbols are untouched by node 1 and $1\in\Lmsc$, the distributions must match. Furthermore, because neither $e$ nor $f$ are touched by the traitor node 2, $(\tilde{E},\tilde{F})\sim(E,F)$. To summarize:
\begin{align}
(\tilde{D},\tilde{E})&\sim(D,E),\\
(\tilde{D},\tilde{F})&\sim(D,F),\\
(\tilde{E},\tilde{F})&\sim(E,F).
\end{align}
Using these three conditions, we apply Property~\ref{prop:pairmagic} to conclude that $(\tilde{D},\tilde{E},\tilde{F})\sim(D,E,F)$. We may do this because, as we argued in Sec.~\ref{sec:caterpillar}, Property~\ref{prop:pairmagic} holds for for any three variables in a polytope subject to a single linear constraint with nonzero coefficients on each one. Since we have constructed the 6 variables to be a $(6,2)$ MDS code, this is true here. (In the space defined by \eqref{eq:62mds1}--\eqref{eq:62mds4}, the three variables $D,E,F$ are subject to $D+E-F=0$.) Since $e$ and $f$ together specify the entire message, in order for this three-way distribution to match, the only choice for $d$ is the true value of $d$. Now we have to show that $c$ can also not be corrupted by the traitor. Since the only symbol seen by the destination that could be touched by node 1 is $a$, we must have $(\tilde{C},\tilde{D},\tilde{F})\sim(C,D,F)$, or else 1 would not be in $\Lmsc$. Again since any two symbols specify the entire message, and both $d$ and $f$ are uncorrupted by the traitor, the value for $c$ sent by node 2 must also be its true value. Therefore the destination will not make an error by using $c$ and $d$ to decode.

The above analysis holds for any $\Lmsc$ containing $\{1,2\}$. That is, if node 2 is the traitor, and $1\in\Lmsc$, then node 2 cannot corrupt $c$ or $d$ (even if $\Lmsc$ contains additional nodes). To prove correctness of the code, it is enough to demonstrate a similar fact for every pair of nodes: we must show that for every pair of nodes $(i,j)$, if $i$ is the traitor and $j\in\Lmsc$, node $i$ is forced to transmit the true value of any symbol that is not also touched by node $j$. If this can be shown for each pair, the destination always decodes correctly by discarding only the symbols touched by every node in $\Lmsc$.

Moreover, it is enough to consider each unordered pair only once. For example, as we have already performed the analysis for $i=2$ and $j=1$, we do not need to perform the same analysis for $i=1$ and $j=2$. This is justified as follows. We have shown that when node 2 is the traitor and $1\in\Lmsc$, symbols $c$ and $d$ are uncorrupted. Therefore $(\tA,\tC,\tD,\tF)\sim(A,C,D,F)$. Hence if $1\in\Lmsc$ and $(\tA,\tC,\tD,\tF)\not\sim(A,C,D,F)$, node 2 cannot be the traitor, so $2\notin\Lmsc$. Now, if node 1 is the traitor and $2\in\Lmsc$, then it must be the case that $(\tA,\tC,\tD,\tF)\sim(A,C,D,F)$. Since of these four symbols only $a$ is touched by node 1, it cannot be corrupted. This same argument can apply to any pair of nodes.

We now complete the proof of correctness of the proposed Polytope Code for the Cockroach network by considering all pairs of potential traitors in the network:
\begin{description}
\item[$(1,2)$:] Proof above.
\item[$(1,3)$:] Suppose node 1 is the traitor and $3\in\Lmsc$. We must show that node 1 cannot corrupt $a$. We have that $(\tA,\tC,\tD)\sim(A,C,D)$, because these three symbols are not touched by node 3, and are available at the destination. Since $c$ and $d$ determine the message, this single constraint is enough to conclude that node 1 cannot corrupt $a$. This illustrates a more general principle: when considering the pair of nodes $(i,j)$, if the number of symbols available at the destination untouched by both $i$ or $j$ is at least as large as the rate of the code, we may immediately conclude that no symbols can be corrupted. In fact, this principle works even for finite-field linear codes.
\item[$(1,4)$:] Follows exactly as $(1,3)$.
\item[$(1,5)$:] Follows exactly as $(1,3)$.
\item[$(2,3)$:] Follows exactly as $(1,2)$.
\item[$(2,4)$:] Suppose node 4 is the traitor and $2\in\Lmsc$. The only symbol touched by both nodes 1 and 4 is $c$, so the destination will decode from $a,d,f$. But node 4 does not touch any of these symbols, so it cannot corrupt them.
\item[$(2,5)$:] Follows exactly as $(2,4)$.
\item[$(3,4)$:] Follows exactly as $(1,3)$.
\item[$(3,5)$:] Follows exactly as $(1,3)$.
\item[$(4,5)$:] Follows exactly as $(1,3)$.
\end{description}

\section{The Polytope Code}\label{sec:polytope}

We now describe the general structure of Polytope Codes and state their important properties. Given a matrix $F\in\Zmbb^{u\times m}$, consider the polytope
\beq \Pmsc_k=\big\{\xbf\in\Zmbb^m:F\xbf=0,|x_i|\le k\text{ for }i=1,\ldots,m\big\}.\label{eq:polytope}\eeq
We may also describe this polytope in terms of a matrix $K$ whose columns form a basis for the null-space of $F$. Let $\Xbf$ be an $m$-dimensional random vector uniformly distributed over $\Pmsc_k$. Take $n$ to be a multiple of the least common denominator of the distribution of $\Xbf$ and let $T^n(\textbf{X})$ be the set of sequences $\xbf^n$ with joint type exactly equal to this distribution. In a Polytope Code, each message is associated with an element of $T^n(\textbf{X})$. By the theory of types, the number of elements in this set is at least $2^{n(H(\Xbf)-\eps)}$ for any $\eps>0$ and sufficiently large $n$. Given a message and the corresponding sequence $\xbf^n$, each edge in the network holds a sequence $x_i^n$ for some $i=1,\ldots,m$. As we have seen in the example Polytope Codes in Sec.~\ref{sec:caterpillar} and~\ref{sec:cockroach}, the joint entropies for large $k$ can be calculated just from the properties of the linear subspace defined by $F$. The following proposition states this property in general.
\begin{prop}
\label{prop:entropy}
For any $S\subseteq\{1,\ldots,m\}$
\beq \lim_{k\to\infty}\frac{H(X_S)}{\log k}=\rank(K_S)\label{eq:polyentropy}\eeq
where $K_S$ is the matrix made up of the rows of $K$ corresponding to the elements of $S$.
\end{prop}
\begin{proof}
For any $S\subseteq\{1,\ldots,m\}$, let $\Pmsc_k(X_S)$ be the projection of $\Pmsc_k$ onto the subspace made up of dimensions $S$. The number of elements in $\Pmsc_k$ is $\Theta(k^{\rank(K_S)})$. That is, there exist constants $c_1$ and $c_2$ such that for sufficiently large $k$
\beq c_1 k^{\rank(K_S)}\le |\Pmsc_k(X_S)|\le c_2 k^{\rank(K_S)}.\label{eq:cbounds}\eeq
For $S=\{1,\ldots,m\}$, because $\Xbf$ is defined to be uniform on $\Pmsc_k$, \eqref{eq:cbounds} gives
\beq\lim_{k\to\infty}\frac{H(\Xbf)}{\log k} = \lim_{k\to\infty}\frac{\log|\Pmsc_k|}{\log k}=\rank(K).\eeq
Moreover, by the uniform bound
\beq\lim_{k\to\infty}\frac{H(X_S)}{\log k}\le \rank(K_S).\label{eq:ent1}\eeq
For any $S\subset\{1,\ldots,m\}$, let $T\subset\{1,\ldots,m\}$ be a minimal set of elements such that $\rank(K_{S,T})=\rank(K)$; i.e. such that $X_{S,T}$ completely specify $X$ under the constraint $FX=0$. Note that $\rank(K_T)=\rank(K)-\rank(K_S)$. Hence
\begin{align}
\lim_{k\to\infty}\frac{H(X_S)}{\log k} &= \lim_{k\to\infty}\frac{H(X_{S,T})}{\log k} - \frac{H(X_T|X_S)}{\log k}\\
&\ge\lim_{k\to\infty}\frac{H(X)}{\log k} - \frac{H(X_T)}{\log k}\\
&\ge\rank(K)-\rank(K_T)\\
&=\rank(K_S).\label{eq:ent2}\end{align}
Combining \eqref{eq:ent1} with \eqref{eq:ent2} completes the proof
\end{proof}

Recall that in a linear code operating over the finite field $\Fmbb$, we may express the elements on the edges in a network $\xbf\in\Fmbb^m$ as a linear combination of the message $\xbf=K\wbf$, where $K$ is a linear transformation over the finite field, and $\wbf$ is the message vector. Taking a uniform distribution on $\wbf$ imposes a distribution on $\Xbf$ satisfying
\beq H(X_S)=\rank(K_S)\log|\Fmbb|.\eeq
This differs from \eqref{eq:polyentropy} only by a constant factor, and also that \eqref{eq:polyentropy} holds only in the limit of large $k$. Hence, Polytope Codes achieve a similar set of entropy profiles as standard linear codes. They may not be identical, because interpreting a matrix $K_S$ as having integer values as opposed to values from a finite field may cause its rank to change. However, the rank when interpreted as having integer values can never be less than when interpreted as having finite field values, because any linear equality on the integers will hold on a finite field, but not necessarily vice versa. The matrix $K_S$ could represent, for example, the source-to-destination linear transformation in a code, so its rank is exactly the achieved rate. Therefore, a Polytope Code always achieves at least as high a rate as the equivalent linear code. Often, when designing linear codes, the field size must be made sufficiently large before the code works; here, sending $k$ to infinity serves much the same purpose, albiet only asymptotically.

In Sec.~\ref{sec:caterpillar} and~\ref{sec:cockroach}, we saw that Property~\ref{prop:pairmagic} played an important role in the functionality of the Polytope Codes. The following theorem states the more general version of this property. It compromises the major property that Polytope Codes possess and linear codes do not.

\begin{theorem}[Fundamental Property of Polytope Codes]\label{thm:magic}
Let $\Xbf\in\Rmbb^m$ be a random vector satisfying $F\Xbf=0$. Suppose a second random vector $\tilde{\Xbf}\in\Rmbb^m$ satisfies the following $L$ constraints:
\beq A_l\tilde{\Xbf}\sim A_l\Xbf\text{ for }l=1,\ldots,L\label{eq:pqcon}\eeq
where $A_l\in\Rmbb^{u_l\times m}$. The two vectors are equal in distribution if the following hold:
\begin{enumerate}
\item There exists a positive definite $C\in\Rmbb^{u\times u}$ and matrices $\Sigma_l\in\Rmbb^{u_l\times u_l}$ such that
\beq F^TCF = \sum_{l=1}^L A_l^T\Sigma_lA_l\label{eq:hsigma}.\eeq
\item There exists $l^*\in\{1,\ldots,L\}$ such that $\genfrac{[}{]}{0pt}{0}{F}{A_{l^*}}$ has full column rank.
\end{enumerate}
\end{theorem}

\begin{proof}
The following proof follows almost exactly the same argument as the proof of Property~\ref{prop:pairmagic} in Sec.~\ref{sec:caterpillar}. We may write
\begin{align}
\Embb\big[(F\tilde{\Xbf})^TC(F\tilde{\Xbf})\big]
&=\sum_{l=1}^m\Embb\big[(A_l\tilde{\Xbf})^T\Sigma_l(A_l\tilde{\Xbf})\big]\label{eq:fundpf1}\\
&=\sum_{l=1}^m\Embb\big[(A_l\Xbf)^T\Sigma_l(A_l\Xbf)\big]\label{eq:fundpf2}\\
&=\Embb\big[(F\Xbf)^TC(F\Xbf)\big]\label{eq:fundpf3}\\
&=0\label{eq:fundpf4}
\end{align}
where \eqref{eq:fundpf1} and \eqref{eq:fundpf3} follow from \eqref{eq:hsigma}; \eqref{eq:fundpf2} follows from \eqref{eq:pqcon}, and because each term in the sum involves $A_l\Xbf$ for some $l$; and \eqref{eq:fundpf4} follows because $F\Xbf=0$. Because $C$ is positive definite, \eqref{eq:fundpf4} implies $F\tilde{\Xbf}=0$.

By the second property in the statement of the theorem, there exists $G_1\in\Rmbb^{m\times u}$ and $G_2\in\Rmbb^{m\times u_{l^*}}$ such that
\beq G_1F+G_2A_{l^*}=I.\eeq
Hence $G_2A_{l^*}\tilde{\Xbf}=\tilde{\Xbf}$, so we may write
\begin{align}
\tilde{\Xbf}
&=G_2A_{l^*}\tilde{\Xbf}\\
&\sim G_2A_{l^*}\Xbf\\
&=\Xbf.
\end{align}
\end{proof}

As an example of an application of Theorem~\ref{thm:magic}, we use it to prove again Property~\ref{prop:pairmagic} in Sec.~\ref{sec:caterpillar}. Recall that variables $X,Y,Z\in\{-k,\ldots,k\}$ satisfying $X+Y+Z=0$, and the three pairwise distributions of $\tilde{X},\tilde{Y},\tilde{Z}$ match as stated in \eqref{eq:pq1}--\eqref{eq:pq3}. In terms of the notation of Theorem~\ref{thm:magic}, we have $m=3$, $L=3$, and
\begin{align}
F=\left[\begin{array}{ccc}1&1&1\end{array}\right],\\
A_1=\left[\begin{array}{ccc}1&0&0\\0&1&0\end{array}\right],\\
A_2=\left[\begin{array}{ccc}1&0&0\\0&0&1\end{array}\right],\\
A_3=\left[\begin{array}{ccc}0&1&0\\0&0&1\end{array}\right].
\end{align}

To satisfy the second condition of Theorem~\ref{thm:magic}, we may set $l^*=1$, since $\genfrac{[}{]}{0pt}{0}{F}{A_{1}}$ has rank 3. In fact, we could just as well have set $l^*$ to 2 or 3. To verify the first condition, we need to check that there exist $\Sigma_l$ for $l=1,2,3$ and a positive definite $C$ (in this case, a positive scalar, because $F$ has only one row) satisfying \eqref{eq:hsigma}. If we let
\beq\Sigma_l=\left[\begin{array}{cc}\sigma_{l,11}&\sigma_{l,12}\\ \sigma_{l,21}&\sigma_{l,22}\end{array}\right]\eeq
then, for instance,
\beq A_1^T\Sigma_1A_1 = \left[\begin{array}{ccc}\sigma_{1,11}&\sigma_{1,12}&0\\ \sigma_{1,21}&\sigma_{1,22}&0\\ 0&0&0\end{array}\right].\eeq
The right hand side of \eqref{eq:hsigma} expands to
\beq \sum_{l=1}^3 A_l^T\Sigma_lA_l = \left[\begin{array}{ccc}\sigma_{1,11}+\sigma_{2,11}&\sigma_{1,12}&\sigma_{2,12}\\
\sigma_{1,21}&\sigma_{1,22}+\sigma_{3,11}&\sigma_{3,12}\\
\sigma_{2,21}&\sigma_{3,21}&\sigma_{2,22}+\sigma_{3,22}\end{array}\right].\eeq
Therefore, for suitable choices of $\{\Sigma_l\}_{l=1}^3$, we can produce any matrix for the right hand side of \eqref{eq:hsigma}. We may simply set $C=1$ and calculate the resulting matrix for the left hand side, then set $\{\Sigma_l\}_{l=1}^3$ appropriately. This allows us to apply Theorem~\ref{thm:magic} to conclude that $(\tX,\tY,\tZ)\sim(X,Y,Z)$.

In our proof of Theorem~\ref{thm:planar}, we will not use Theorem~\ref{thm:magic} in its most general form. Instead, we state three corollaries that will be more convenient. The first is a generalization of the above argument for more than three variables.
\begin{corollary}\label{cor:magic1}
Let $\Xbf$ satisfy $F\Xbf=0$ for some $F\in\Zmbb^{1\times m}$ with all nonzero values. If $\tilde{\Xbf}$ satisfies
\begin{align}
(\tilde{X_i},\tilde{X_j})\sim(X_i,X_j)\text{ for all }i,j=1,\ldots,m\label{eq:ex1pairs}\\
(\tilde{X_2},\cdots,\tilde{X_m})\sim(X_2,\cdots,X_m)\label{eq:ex1big}
\end{align}
then $\tilde{\Xbf}\sim\Xbf$.
\end{corollary}
\begin{proof}
We omit the explicit construction of the $A_l$ matrices corresponding to the conditions \eqref{eq:ex1pairs}, \eqref{eq:ex1big}. The second condition for Theorem~\ref{thm:magic} is satisfied by \eqref{eq:ex1big}, since the linear constraint $F\Xbf=0$ determines $X_1$ given $X_2\cdots X_m$. To verify the first condition, note that from the conditions in \eqref{eq:ex1pairs}, we may construct an arbitrary matrix on the right hand side of \eqref{eq:hsigma} for suitable $\{\Sigma_l\}_{l=1}^L$. Therefore we may simply set $C=1$.
\end{proof}

Corollary~\ref{cor:magic1} considers the case with $m$ variables and $m-1$ degrees of freedom; i.e. a single linear constraint. The following corollary considers a case with $m$ variables and $m-2$ degrees of freedom.

\begin{figure}
\centerline{\begin{psfrags}
\psfrag{x1}[c]{$\tX_1$}
\psfrag{x2}[c]{$\tX_2$}
\psfrag{y1}[c]{$\tX_3$}
\psfrag{y2}[c]{$\tX_4$}
\psfrag{z}[c]{$\tZbf$}
\includegraphics[width=.25\textwidth]{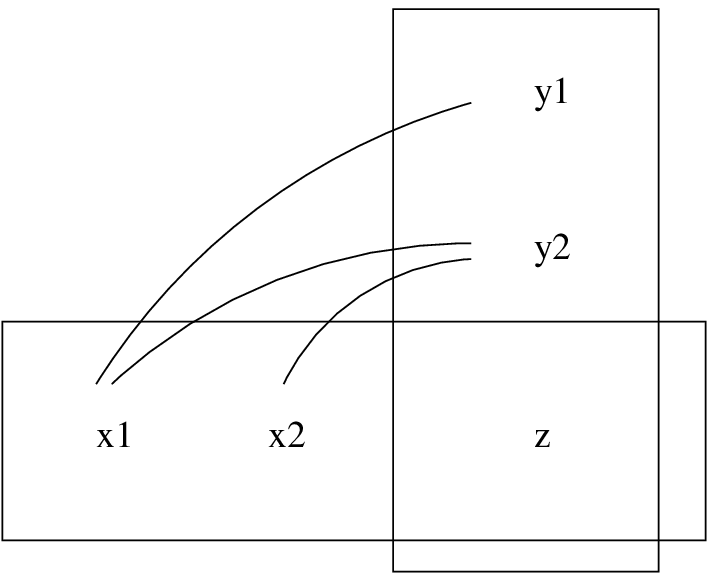}
\end{psfrags}
\hspace{.5in}
\begin{psfrags}
\psfrag{x1}[c]{$\tX_1$}
\psfrag{x2}[c]{$\tX_2$}
\psfrag{y1}[c]{$\tX_3$}
\psfrag{y2}[c]{$\tX_4$}
\psfrag{z}[c]{$\tZbf$}
\includegraphics[width=.25\textwidth]{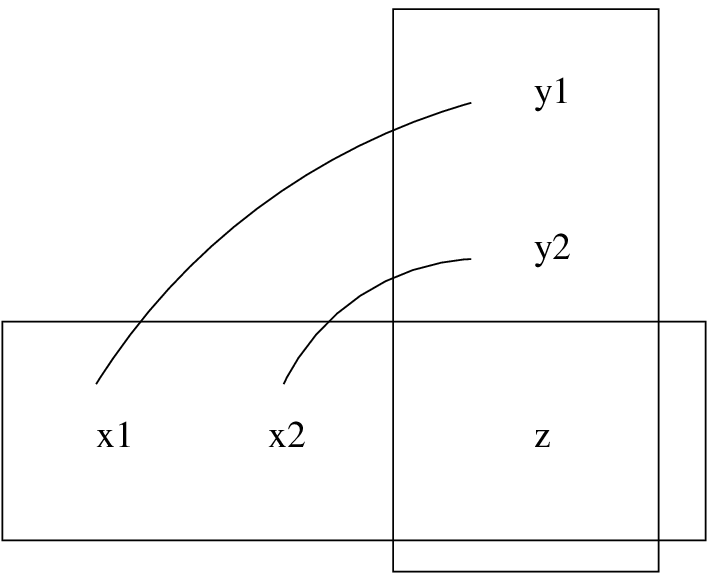}
\end{psfrags}
}
\caption[Diagram of constraints in Corollaries~\ref{cor:magic2} and~\ref{cor:magic3}]{The constraints on the random vector $\tXbf$ in Corollaries~\ref{cor:magic2} (left) and~\ref{cor:magic3} (right). Rectangles represent a constraint on the marginal distribution of all enclosed variables; lines represent pairwise constraints on the two connected variables.}
\label{fig:cormagic}
\end{figure}

\begin{corollary}\label{cor:magic2}
Let $F\in\Zmbb^{2\times m}$ be such that any $2\times 2$ submatrix of $F$ is non-singular. Let $\Xbf$ satisfy $F\Xbf=0$. The non-singular condition on $F$ implies that any $m-2$ variables specify the other two. Assume that $m\ge 4$, and for convenience let $\Zbf=(X_5,\ldots,X_m)$ and $\tZbf=(\tX_5,\ldots,\tX_m)$. If $\tilde{\Xbf}$ satisfies
\begin{align}
(\tX_1,\tX_2,\tZbf)&\sim(X_1,X_2,\Zbf),\label{eq:distcon1}\\
(\tX_3,\tX_4,\tZbf)&\sim(X_3,X_4,\Zbf),\label{eq:distcon2}\\
(\tX_1,\tX_3)&\sim(X_1,X_3),\label{eq:distcon3}\\
(\tX_2,\tX_4)&\sim(X_2,X_4),\label{eq:distcon4}\\
(\tX_1,\tX_4)&\sim(X_1,X_4)\label{eq:distcon5}
\end{align}
then $\tXbf\sim\Xbf$. Fig.~\ref{fig:cormagic} diagrams the constraints on $\tXbf$.
\end{corollary}
\begin{proof}
We prove Corollary~\ref{cor:magic2} with two applications of Corollary~\ref{cor:magic1}. First, consider the group of variables $(X_1X_2X_4\Zbf)$. These $m-1$ variables are subject to a single linear constraint, as in Corollary~\ref{cor:magic1}. From \eqref{eq:distcon1}, \eqref{eq:distcon4}, and \eqref{eq:distcon5} we have all pairwise marginal constraints, satisfying \eqref{eq:ex1pairs}. Furthermore, \eqref{eq:distcon1} satisfies \eqref{eq:ex1big}. We may therefore apply Corollary~\ref{cor:magic1} to conclude
\beq(\tX_1,\tX_2,\tX_4,\tZbf)\sim(X_1,X_2,X_4,\Zbf).\label{eq:magic2a}\eeq
A similar application of Corollary~\ref{cor:magic1} using \eqref{eq:distcon2}, \eqref{eq:distcon3}, and \eqref{eq:distcon5} allows us to conclude
\beq(\tX_1,\tX_3,\tX_4,\tZbf)\sim(X_1,X_3,X_4,\Zbf).\label{eq:magic2b}\eeq
Observe that \eqref{eq:magic2a} and \eqref{eq:magic2b} share the $m$ variables $(\tX_1,\tX_4,\tZbf)$, which together determine $\tX_2$ and $\tX_3$ in exactly the same way that $(X_1,X_4,\Zbf)$ determine $X_2$ and $X_3$. Therefore we may combine \eqref{eq:magic2a} and \eqref{eq:magic2b} to conclude $\tXbf\sim\Xbf$.
\end{proof}

All five constraints \eqref{eq:distcon1}--\eqref{eq:distcon5} are not always necessary, and we may sometimes apply Theorem~\ref{thm:magic} without \eqref{eq:distcon5}. However, this depends on an interesting additional property of the linear constraint matrix $F$, as stated in the third and final corollary to Theorem~\ref{thm:magic}.

\begin{corollary}\label{cor:magic3}
Let $F\in\Zmbb^{2\times m}$ be such that any $2\times 2$ submatrix of $F$ is non-singular, and let $\Xbf$ satisfy $F\Xbf=0$. In addition, assume
\beq|K_{X_1X_2\Zbf}|\ |K_{X_3X_4\Zbf}|\ |K_{X_1X_3\Zbf}|\ |K_{X_2X_4\Zbf}|<0\label{eq:det}\eeq
where again $K$ is a basis for the null space of $F$, and $K_{\Xbf_S}$ for $S\subset\{1,\ldots,m\}$ is the matrix made up of the rows of $K$ corresponding to the variables $(X_i)_{i\in S}$. If $\tXbf$ satisfies \eqref{eq:distcon1}--\eqref{eq:distcon4} (Fig.~\ref{fig:cormagic} diagrams these constraints), then $\tXbf\sim\Xbf$.
\end{corollary}
\begin{proof}
Either \eqref{eq:distcon1} or \eqref{eq:distcon2} satisfies the second condition in Theorem~\ref{thm:magic}. To verify the first condition, first let $G=\sum_l A_l^T\Sigma_lA_l$. In the four constraints \eqref{eq:distcon1}--\eqref{eq:distcon4}, each pair of variables appears together except for $(X_1,X_4)$ and $(X_2,X_3)$. Therefore, for suitable choices of $\Sigma_l$, we can construct any $G$ satisfying $G_{1,4}=G_{2,3}=G_{3,2}=G_{4,1}=0$. We must show that such a $G$ exists satisfying
\beq F^TCF=G\label{eq:fgcondition}\eeq
 for some positive definite $C$.

We build $G$ row-by-row. By \eqref{eq:fgcondition}, each row of $G$ is a linear combination of rows of $F$; i.e. it forms the coefficients of a linear equality constraint imposed on the random vector $\Xbf$. Since $G_{1,4}$, the first row of $G$ represents a linear constraint on the variables $X_1,X_2,X_3,\Zbf$. Since any $m-2$ variables specify the other two, there is exactly one linear equality constraint on these $m-1$ variables, up to a constant. This constraint can be written as
\beq\left|\begin{array}{cc}X_1&K_{X_1}\\X_2&K_{X_2}\\X_3&K_{X_3}\\\Zbf&K_\Zbf\end{array}\right|=0.\label{eq:detform}\eeq
since the vector $X_1,X_2,X_3,\Zbf$ forms a linear combination of the columns of $K_{X_1,X_2,X_3,\Zbf}$. Hence, the first row of $G$ is a constant multiple of the coefficients in \eqref{eq:detform}. In particular,
\begin{align}
G_{1,1}&=\alpha |K_{X_2X_3\Zbf}|,\\
G_{1,2}&=-\alpha |K_{X_1X_3\Zbf}|\label{eq:H12}
\end{align}
for some constant $\alpha$. Since $G_{2,3}=0$, the second row of $G$ represents the linear constraint on $X_1,X_2,X_4,\Zbf$. Using similar reasoning as above gives
\begin{align}
G_{2,1}&=\beta |K_{X_2X_4\Zbf}|,\label{eq:H21}\\
G_{2,2}&=-\beta |K_{X_1X_4\Zbf}|
\end{align}
for some constant $\beta$. Moreover, by \eqref{eq:fgcondition} $G$ is symmetric, so $G_{1,2}=G_{2,1}$, and by \eqref{eq:H12} and \eqref{eq:H21}
\beq \beta = -\frac{|K_{X_1X_3\Zbf}|}{|K_{X_2X_4\Zbf}|}\alpha.\eeq
Positive definiteness of $C$ is equivalent to positive definiteness of the upper left $2\times 2$ block of $G$, so the conditions we need are
\begin{align}
0&<G_{1,1} = \alpha |K_{X_2X_3\Zbf}|,\label{eq:det11}\\
0&<G_{1,1}G_{2,2}-G_{1,2}G_{2,1} \\
&= \alpha^2\left[\frac{|K_{X_2X_3\Zbf}|\ |K_{X_1X_4\Zbf}|\ |K_{X_1X_3\Zbf}|}{|K_{X_2X_4\Zbf}|}
 - |K_{X_1X_3\Zbf}|^2\right].\label{eq:det22}
\end{align}
We may choose $\alpha$ to trivially satisfy \eqref{eq:det11}, and \eqref{eq:det22} is equivalent to
\beq |K_{X_1X_3\Zbf}|\ |K_{X_2X_4\Zbf}|\Big(|K_{X_2X_3\Zbf}|\ |K_{X_1X_4\Zbf}|  -|K_{X_2X_4\Zbf}|\ |K_{X_1X_3\Zbf}|\Big)>0\eeq
which may also be written as \eqref{eq:det}.
\end{proof}

The necessity of satisfying \eqref{eq:det} in order to apply Theorem~\ref{thm:magic} substantially complicates code design. When building a linear code, one need only worry about the rank of certain matrices; i.e. certain determinants need be nonzero. Here, we see that the signs of these determinants may be constrained as well.

\section{Proof of Theorem~\ref{thm:planar}}\label{sec:planarpf}

To prove Theorem~\ref{thm:planar}, we need to specify a Polytope Code for each network satisfying conditions 1--3 in the statement of the theorem. This involves specifying the linear relationships between various symbols in the network, the comparisons that are done among them at internal nodes, and then how the destination uses the comparison information it receives to decode. We then proceed to prove that the destination always decodes correctly. The key observation in the proof is that the important comparisons that go on inside the network are those that involve a variable that does not reach the destination. This is because those symbols that do reach the destination can be examined there, so further comparisons inside the network do not add anything. Therefore we will carefully route these non-destination symbols to maximize the utility of their comparisons. In particular, we design these paths so that for every node having one direct edge to the destination and one other output edge, the output edge not going to the destination holds a non-destination variable. The advantage of this is that any variable, before exiting the network, is guaranteed to cross a non-destination variable at a node where the two variables may be compared. The existence of non-destination paths with this property depends on the planarity of the network. This is described in much more detail in the sequel.

\emph{Notation:} For an edge $e\in E$, with $e=(i,j)$, where $i,j\in V$, let $\head(e)=i$ and $\tail(e)=j$. For a node $i\in V$, let $\Emsc_\inn(i)$ be the set of edges $e$ with $\tail(e)=i$, and let $\Emsc_\out(i)$ be the set of edges $e$ with $\head(e)=i$. Let $\Nmsc_\inn(i)$ be the set of input neighbors of $i$; that is, the set of $\head(e)$ for each $e\in\Emsc_\inn(i)$. Similarly, let $\Nmsc_\out(i)$ be the set of output neighbors of $i$. For integers $a,b$, let $\Vmsc_{a,b}$ be the set of nodes with $a$ inputs and $b$ outputs. We will sometimes refer to such nodes as $a$-to-$b$. For $l\in\{1,2\}$, let $\bar{l}=2-l$. A \emph{path} is defined as an ordered list of edges $e_1,\ldots,e_k$ satisfying $\tail(e_l)=\head(e_{l+1})$ for $l=1,\ldots,k-1$. The head and tail of a path are defined as $\head(e_1)$ and $\tail(e_k)$ respectively. A node $i$ is said to \emph{reach} a node $j$ if there exists a path with head $i$ and tail $j$. By convention, a node can reach itself.

Consider an arbitrary network satisfying the conditions of Theorem~\ref{thm:planar}. By condition (3), no node has more output edges than input edges. Therefore the min-cut is that between the destination and the rest of the network. Let $M$ be the value of this cut; i.e., the number of edges connected to the destination. We now state a lemma giving instances of the cut-set upper bound on capacity in terms of quantities that make the bound easier to handle than Theorem~\ref{thm:cutset} itself. We will subsequently show that the minimum upper bound given by Lemma~\ref{lemma:cutset2net} is achievable using a Polytope Code; therefore, the cut-set bound gives the capacity.

\begin{lemma}\label{lemma:cutset2net}
For $i,j\in V$, let $d_{i,j}$ be the sum of $|\Emsc_\inn(k)|-|\Emsc_\out(k)|$ for all nodes $k$ reachable from either $i$ or $j$, not including $i$ or $j$. That is, if $k$ is $a$-to-$b$, it contributes $a-b$ to the sum. Recall that this difference is always positive. Let $c_i$ be the total number of output edges from node $i$, and let $e_i$ be the number of output edges from node $i$ that go directly to the destination.
For any distinct pair of nodes $i_1,i_2$,
\beq C\le M-e_{i_1}-e_{i_2}.\label{eq:cutset2netA}\eeq
Moreover, if there is no path between $i_1$ and $i_2$,
\beq C\le M+d_{i_1,i_2}-c_{i_1}-c_{i_2}.\label{eq:cutset2netB}\eeq
\end{lemma}
\begin{proof}
Applying Theorem~\ref{thm:cutset} with $A=V\setminus\{D\}$, $T=\{i_1,i_2\}$ immediately gives \eqref{eq:cutset2netA}. To prove \eqref{eq:cutset2netB}, we apply Theorem~\ref{thm:cutset} with $T=\{i_1,i_2\}$, and
\beq A=\{k\in V:k\text{ is not reachable from $i_1$ or $i_2$}\}\cup\{i_1,i_2\}.\eeq
Observe that there are no backwards edges for the cut $A$, because any node in $A^c$ is reachable from either $i_1$ or $i_2$, so for any edge $(j,k)$ with $j\in A^c$, $k$ is also reachable by from $i_1$ or $i_2$, so $k$ is also not in $A$. Therefore we may apply Theorem~\ref{thm:cutset}. Since all output neighbors of $i_1$ and $i_2$ are not in $A$, each output edge of $i_1$ and $i_2$ crosses the cut. Hence \eqref{eq:cutset} becomes
\beq C\le|\{e\in E:\head(e)\in A,\ \tail(e)\notin A\}|-c_1-c_2.\label{eq:cbound1}\eeq
Since no node in the network has more output edges than input edges, the difference between the first term in \eqref{eq:cbound1}---the number of edges crossing the cut---and $M$ is exactly the sum of $|\Emsc_\inn(k)|-|\Emsc_\out(k)|$ for all $k\in A^c$. Hence
\beq|\{e\in E:\head(e)\in A,\ \tail(e)\notin A\}|-M=d_{i_1,i_2}.\label{eq:eamdiff}\eeq
Combining \eqref{eq:cbound1} with \eqref{eq:eamdiff} gives \eqref{eq:cutset2netB}.
\end{proof}

Next, we show that we may transform any network satisfying the conditions of Theorem~\ref{thm:planar} into an equivalent one that is planar, and made up of just 2-to-2 nodes and 2-to-1 nodes. We will go on to show that the upper bound provided by Lemma~\ref{lemma:cutset2net} is achievable for any such network, so it will be enough to prove that a transformation exists that preserves planarity, does not reduce capacity, and does not change the bound given by Lemma~\ref{lemma:cutset2net}.

We first replace any $a$-to-$b$ node $i$ with a cascade of $a-b$ 2-to-1 nodes followed by a $b$-to-$b$ node. This transformation is illustrated in Fig.~\ref{fig:transform}. Denote the $b$-to-$b$ node in the transformation $i^*$. Since no node in the original network has more than two output edges, the resulting network contains only 1-to-1 nodes, 2-to-2 nodes, and 2-to-1 nodes. We will shortly argue that the 1-to-1 nodes may be removed as well. Certainly these transformations maintain the planarity of the network. Moreover, any rate achievable on the transformed network is also achievable on the original network. This is because if node $i$ is transformed via this operation into several nodes, any coding operation performed by these nodes can certainly be performed by node $i$. Additionally, the traitor taking control of node $i$ in the original network does exactly as much damage as the traitor taking control of $i^*$ in the transformed network, since it controls all edges sent to other nodes. Now consider the minimum upper bound given by Lemma~\ref{lemma:cutset2net} after this transformation. The only nodes with positive $e_j$ values will be $i^*$ nodes, and $e_{i^*}=e_i$. Hence \eqref{eq:cutset2netA} cannot change. In \eqref{eq:cutset2netB}, if we take $i_1^*$ and $i_2^*$, then the bound is the same in the transformed network. Taking one of the 2-to-1 nodes instead of a $i^*$ node cannot result in a lower bound, because they have no more output edges, so no higher $c$ values, and no fewer reachable nodes with fewer outputs than inputs, so no smaller $d$ values. Therefore, the minimal bound given by \eqref{eq:cutset2netB} for the transformed network is the same as that of the original network. Moreover, in the transformed network $d_{i_1,i_2}$ is equal simply to the number of 2-to-1 nodes reachable from $i_1$ or $i_2$ not including $i_1,i_2$.

\begin{figure}
\begin{psfrags}
\psfrag{i}[][][.65]{$i$}
\psfrag{is}[][][.65]{$i^*$}
\centerline{\includegraphics[width=3.3in]{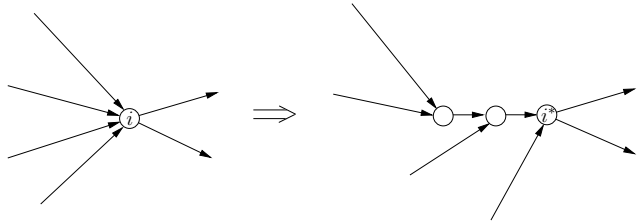}}
\end{psfrags}
\caption[Transformation of a 4-to-2 node]{An illustration of the transformation from a 4-to-2 node to an equivalent set of 2-to-1 and 2-to-2 nodes.}
\label{fig:transform}
\end{figure}

We may additionally transform the network to remove 1-to-1 nodes, simply be replacing the node and the two edges connected to it by a single edge. The traitor can always take over the preceding or subsequent node and have at least as much power. The only exception is when the 1-to-1 node is connected only to the source and destination. In this case, instead of removing the node, we may add a additional edge to it from the source, turning it into a 2-to-1 node. Such a transformation does not change the capacity, nor the planarity or the Lemma~\ref{lemma:cutset2net} bounds.

We also assume without loss of generality that all nodes in the network are reachable from the source. Certainly edges out of these nodes cannot carry any information about the message, so we may simply discard this portion of the network, if it exists, without changing the capacity.

We will show that the smallest bound given by Lemma~\ref{lemma:cutset2net} is achievable using a Polytope Code. If we take $i_1$ and $i_2$ to be two nodes with at least one direct link to the destination, \eqref{eq:cutset2netA} gives that the capacity is no more than $M-2$. Moreover, since $e_i\le c_i\le 2$ for any node $i$, neither \eqref{eq:cutset2netA} nor \eqref{eq:cutset2netB} can produce a bound less than $M-4$. Therefore the minimum bound given by Lemma~\ref{lemma:cutset2net} can take on only three possible values: $M-4,M-3,M-2$. It is not hard to see that $M-4$ is trivial achievable; indeed, even with a linear code. Therefore the only interesting cases are when the cut-set bound is $M-3$ or $M-2$. We begin with the latter, because the proof is more involved, and contains all the necessary parts to prove the $M-3$ case. The $M-3$ proof is subsequently given in Section~\ref{subsection:m3proof}.

Assume that the right hand sides of \eqref{eq:cutset2netA} and \eqref{eq:cutset2netB} are never smaller than $M-2$. We describe the construction of the Polytope Code to achieve rate $M-2$ in several steps. The correctness of the code will be proved in Lemmas~\ref{lemma:labels}--\ref{lemma:pairwise}, which are stated during the description of the construction process. These Lemmas are then proved in Sections~\ref{subsection:labels}--\ref{subsection:pairwise}.

\emph{1) Edge Labeling:} We first label all the edges in the network except those in $\Emsc_{\inn}(D)$. These labels are denoted by the following functions
\begin{align}
\phi:E\setminus\Emsc_{\inn}(D)&\to\Vmsc_{2,1}\label{eq:phidef}\\
\psi:E\setminus\Emsc_{\inn}(D)&\to\{0,1\}.\label{eq:psidef}
\end{align}
For a 2-to-1 node $v$, let $\Lambda(v)$ be the set of edges $e$ with $\phi(e)=v$. The set $\Lambda(v)$ represents the edges carrying symbols that interact with the non-destination symbol that terminates at node $v$. The set of edges with $\phi(e)=v$ and $\psi(e)=1$ represent the path taken by the non-destination symbol that terminates at node $v$. The following Lemma states the existence of labels $\phi,\psi$ with the necessary properties.
\begin{lemma}\label{lemma:labels}
There exist functions $\phi$ and $\psi$ with the following properties:
\begin{description}
\item[\textbf{A}] The set of edges $e$ with $\phi(e)=v$ and $\psi(e)=1$ form a path.
\item[\textbf{B}] If $\phi(e)=v$, then either $\tail(e)=v$ or there is an edge $e'$ with $\head(e')=\tail(e)$ and $\phi(e')=v$.
\item[\textbf{C}] For every 2-to-2 node $i$ with output edges $e_1,e_2$, either $\psi(e_1)=1$, $\psi(e_2)=1$, or $\phi(e_1)\neq\phi(e_2)$.
\end{description}
\end{lemma}
Note that if property (B) holds, $\Lambda(v)$ is a union of paths ending at $v$. From property (A), the edges on one of these paths satisfy $\psi(e)=1$.

\emph{2) Internal Node Operation:} Assume that $\phi$ and $\psi$ are defined to satisfy properties (A)--(C) in Lemma~\ref{lemma:labels}. Given these labels, we will specify how internal nodes in the network operate. Every edge in the network will hold a symbol representing a linear combination of the message, as well as possibly some comparison bits. We also define a function
\beq\rho:E\to\{1,\ldots,|\Emsc_\out(S)|\}\eeq
that will serve as an accounting tool to track symbols as they pass through the network. We begin by assigning distinct and arbitrary values to $\rho(e)$ for all $e\in\Emsc_\out(S)$ ($\rho$ therefore constitutes an ordering on $\Emsc_\out(S)$). Further assignments of $\rho$ will be made recursively. This will be made explicit below, but if a symbol is merely forwarded, it travels along edges with a constant $\rho$. When linear combinations occur at internal nodes, $\rho$ values are manipulated, and $\rho$ determine exactly how this is done.

For every node $i$ with 2 input edges, let $f_1,f_2$ be these edges. If $i$ is 2-to-2, let $e_1,e_2$ be its two output edges; if it is 2-to-1, let $e$ be its output edge. If $\phi(f_1)=\phi(f_2)$, then node $i$ compares the symbols on $f_1$ and $f_2$. If node $i$ is 2-to-2, then $\phi(e_l)=\phi(f_1)$ for either $l=1$ or $2$. Node $i$ transmits its comparison bit on $e_l$. If node $i$ is 2-to-1, then it transmits its comparison bit on $e$. All 2-to-2 nodes forward all received comparison bits on the output edge with the same $\phi$ value as the input edge on which the bit was received. All 2-to-1 nodes forward all received comparison bits on its output edge.

We divide nodes in $\Vmsc_{2,2}$ into the following sets. The linear transformation performed at node $i$ will depend on which of these sets it is in.
\begin{align}
\Wmsc_1&=\{i\in\Vmsc_{2,2}:\psi(f_1)=\psi(f_2)=0,\phi(f_1)\neq\phi(f_2)\}\\
\Wmsc_2&=\{i\in\Vmsc_{2,2}:\psi(f_1)=\psi(f_2)=0,\phi(f_2)=\phi(f_2)\}\\
\Wmsc_3&=\{i\in\Vmsc_{2,2}:\psi(f_1)=1\text{ or }\psi(f_2)=1\}
\end{align}
We will sometimes refer to nodes in $\Wmsc_2$ as \emph{branch nodes}, since they represent branches in $\Lambda(\phi(f_1))$. Moreover, branch nodes are significant because a failed comparison at a branch node will cause the forwarding pattern within $\Lambda(\phi(f_1))$ to change. For an edge $e$, $X_e$ denotes the symbol transmitted on $e$. The following gives the relationships between these symbols, which are determined by internal nodes, depending partially on the comparison bits they receive. For each node $i$, the action of node $i$ depends on which set it falls in as follows:
\begin{itemize}
\item $\Wmsc_1$: Let $l$ be such that $\phi(e_l)=\phi(f_1)$. The symbol on $f_1$ is forwarded to $e_l$, and the symbol on $f_2$ is forwarded onto $e_{\bar{l}}$. Set $\rho(e_l)=\rho(f_1)$, and $\rho(e_{\bar{l}})=\rho(f_2)$.
\item$\Wmsc_2$: Let $l$ be such that $\phi(e_l)=\phi(f_1)=\phi(f_2)$. Let $l'$ be such that $\rho(f_{l'})<\rho(f_{\bar{l'}})$. We will show in Lemma~\ref{lemma:tree} that our construction is such that $\rho(f_1)\ne\rho(f_2)$ at all nodes, so $l'$ is well defined. If neither $f_1$ nor $f_2$ hold a failed comparison bit, the output symbols are
    \begin{align}
    X_{e_l}&=\gamma_{i,1}X_{f_1}+\gamma_{i,2}X_{f_2}\label{eq:xel}\\
    X_{e_{\bar{l}}}&=X_{f_{l'}}\label{eq:xebarl}\end{align}
    where coefficients $\gamma_{i,1},\gamma_{i,2}$ are nonzero integers to be chosen later. Set output $\rho$ values to
    \begin{align}
    \rho(e_l)&=\rho(f_{\bar{l'}})\\
    \rho(e_{\bar{l}})&=\rho(f_{l'}).\end{align}
    Note that the symbol on the input edge with smaller $\rho$ value is forwarded without linear combination. If the input edge $f_{l'}$ reports a failed comparison anywhere previously in $\Lambda(\phi(f_1))$, then \eqref{eq:xebarl} changes to
    \beq X_{e_{\bar{l}}}=X_{f_{\bar{l'}}}.\eeq
\item$\Wmsc_3$: Let $l$ be such that $\psi(f_l)=1$, and $l'$ be such that $\psi(e_{l'})=1$ and $\phi(e_{l'})=\phi(f_l)$. The symbol on $f_l$ is forwarded to $e_{l'}$, and the symbol on $f_{\bar{l}}$ is forwarded to $e_{\bar{l}'}$, with the following exception. If $\phi(f_1)=\phi(f_2)$ and there is a failed comparison bit sent from $f_{\bar{l}}$, then the forwarding swaps: the symbol on $f_l$ is forwarded to $e_{\bar{l}'}$, and the symbol on $f_{\bar{l}}$ is forwarded to $e_{l'}$. Set $\rho(e_{l'})=\rho(f_l)$ and $\rho(e_{\bar{l}'})=\rho(f_{\bar{l}})$. Again, $\rho$ is consistent along forwarded symbols, but only when all comparisons succeed.
\item$\Vmsc_{2,1}$: Let $l$ be such that $\psi(f_l)=1$. The symbol from $f_{\bar{l}}$ is forwarded on $e$, unless there is a failed comparison bit sent from $f_{\bar{l}}$, in which case the symbol from $f_l$ is forwarded on $e$. Set $\rho(e)=\rho(f_{\bar{l}})$.
\end{itemize}
See Fig.~\ref{fig:lambdaex} for an illustration of the linear transformations performed at internal nodes and how they change when a comparison fails. The following Lemma gives some properties of the internal network behavior as prescribed above.
\begin{lemma}\label{lemma:tree}
The following hold:
\begin{enumerate}
\item For any integer $a\in\{1,\ldots,|\Emsc_{\out}(S)|\}$, the set of edges with $e$ with $\rho(e)=a$ form a path (we refer to this in the sequel as the \emph{$\rho=a$ path}). Consequently, there is no node $i$ with input edges $f_1,f_2$ such that $\rho(f_1)=\rho(f_2)$.
\item If there are no failed comparisons that occur in the network, then the linear transformations are such that the decoder can decode any symbol in the network except those on non-destination paths.
\item Suppose a comparison fails at a branch node $k$ with input edges $f_1,f_2$ with $v=\phi(f_1)=\phi(f_2)$. Assume without lack of generality that $\rho(f_1)<\rho(f_2)$. The forwarding pattern within $\Lambda(v)$ changes such that symbols sent along the $\rho=\rho(f_2)$ path are not decodable at the destination, but what was the non-destination symbol associated with $v$ is decodable.
\end{enumerate}
\end{lemma}

\newcommand{\pscc}{.6}

\begin{figure}
\begin{psfrags}
\psfrag{a}[][][\pscc]{$a$}
\psfrag{b}[][][\pscc]{$b$}
\psfrag{c}[][][\pscc]{$c$}
\psfrag{d}[][][\pscc]{$d$}
\psfrag{e}[][][\pscc]{$e$}
\psfrag{f}[][][\pscc]{$f$}
\psfrag{g}[][][\pscc]{$g$}
\psfrag{h}[][][\pscc]{$h$}
\psfrag{i}[][][\pscc]{$i$}
\psfrag{ip}[][][\pscc]{$i\ [d,e,f]$}
\psfrag{ab}[][][\pscc]{$a,b$}
\psfrag{abc}[][][\pscc]{$a,b,c$}
\psfrag{de}[][][\pscc]{$d,e$}
\psfrag{dep}[][][\pscc]{$d,e\ [f]$}
\psfrag{def}[][][\pscc]{$d,e,f$}
\psfrag{defp}[][][\pscc]{$d,e,f\ [i]$}
\psfrag{defg}[][][\pscc]{$d,e,f,g$}
\psfrag{v}[][][.55]{$v$}
\centerline{\includegraphics[width=3.3in]{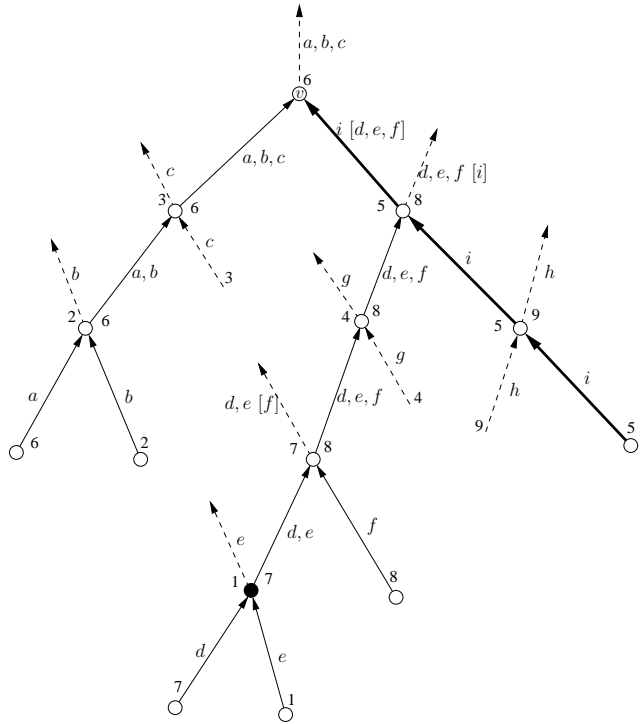}}
\end{psfrags}
\caption[An example of the linear transformations performed in the subnetwork made up of all edges with the same label]{An example of the linear transformations performed in $\Lambda(v)$ for some $v$ (labeled as such). Solid edges denote $\phi(e)=v$, dashed edges denote $\phi(e)\ne v$. Thick edges denote $\psi(e)=1$. Near the head of each edge is the corresponding $\rho$ value. Also shown is the symbol transmitted along that edge, given initial symbols $a$--$i$ at the furthest upstream edges in the network. When several symbols are written on an edge, this indicates that the edge carries a linear combination of those symbols. The symbols indicated in brackets are those carried by the edges when the comparison at the indicated black node fails. Symbols on edges labeled without brackets do not change when the comparison fails.}
\label{fig:lambdaex}
\end{figure}

\emph{3) MDS Code Construction:} The rules above explain how the symbols are combined and transformed inside the network. In addition, when the initial set of symbols are sent into the network from the source, they are subject to linear constraints. We now describe exactly how this is done. Assume that no comparisons fail in the network, so the linear relationships between symbols are unmodified. For a 2-to-1 node $v$, let $e^*_v$ be the edge with $\phi(e_v^*)=v$, $\psi(e_v^*)=1$, and $\tail(e_v^*)=v$; i.e. it is the last edge to hold the non-destination symbol terminating at $v$. Observe that it will be enough to specify the linear relationships among the symbols on $\{e_v^*:v\in\Vmsc_{2,1}\}$ as well as the $M$ edges in $\Emsc_\inn(D)$. These collectively form the Polytope Code equivalent of a $(M+|\Vmsc_{2,1}|,M-2)$ MDS code. We must construct this code so as to satisfy certain instances of \eqref{eq:det}, so that we may apply Theorem~\ref{thm:magic} as necessary. The following Lemma states the existence of a set of linear relationships among the $M+|\Vmsc_{2,1}|$ variables with the required properties.
\begin{lemma}\label{lemma:vander}
For each 2-to-1 node $v$, let $\Xi(v)$ be the set of edges $e$ with $\tail(e)=D$ such that there is an edge $e'$ with $\tail(e')=\head(e)$, $\phi(e')=v$, and $\psi(e')=1$. That is, the symbol on $e$, just before being sent to the destination, was compared against the non-destination symbol associated with $v$. Note that any edge $e\in\Emsc_\inn(D)$ is contained in $\Xi(v)$ for some 2-to-1 node $v$. There exists a generator matrix $K\in\Zmbb^{M+|\Vmsc_{2,1}|\times M-2}$ where each row is associated with an edge in $\{e_v^*:v\in\Vmsc_{2,1}\}\cup\Emsc_\inn(D)$ such that for all $v_1,v_2\in\Vmsc_{2,1}$ and all $f_1\in\Xi(v_1),f_2\in\Xi(v_2)$, the constraints
\begin{align}
(\tX_{f_1},\tX_{f_2},\tZbf)&\sim(X_{f_1},X_{f_2},\Zbf)\\
(\tX_{e_{v_1}^*},\tX_{e_{v_2}^*},\tZbf)&\sim(X_{e_{v_1}^*},X_{e_{v_2}^*},\Zbf)\\
(\tX_{f_1},\tX_{e_{v_1}^*})&\sim(X_{f_1},X_{e_{v_1}^*})\\
(\tX_{f_2},\tX_{e_{v_2}^*})&\sim(X_{f_2},X_{e_{v_2}^*})\end{align}
imply
\beq (\tX_{f_1},\tX_{f_2},\tX_{e_{v_1}^*},\tX_{e_{v_2}^*}\tZbf)\sim(X_{f_1},X_{f_2},X_{e_{v_1}^*},X_{e_{v_2}^*}\tZbf)\eeq
where
\beq \Zbf=(X_e:e\in\Emsc_\inn(D)\setminus\{f_1,f_2\}).\eeq
\end{lemma}

\emph{4) Decoding Procedure:} To decode, the destination first compiles a list $\Lmsc\subset V$ of which nodes may be the traitor. It does this by taking all its available data: received comparison bits from interior nodes as well as the symbols it has direct access to, and determines whether it is possible for each node, if it were the traitor, to have acted in a way to cause these data to occur. If so, it adds this node to $\Lmsc$. For each node $i$, let $K_{i}$ be the linear transformation from the message vector $\Wbf$ to the symbols on the output edges of node $i$. With a slight abuse of notation, regard $K_D$ represent the symbols on the input edges to $D$ instead. For a set of nodes $S\subset V$, let $K_{D\perp S}$ be a basis for the subspace spanned by $K_D$ orthogonal to
\beq\bigcap_{j\in S}\spn(K_{j\to D}).\eeq
The destination decodes from $K_{D\perp\Lmsc}\Wbf$. If $i$ is the traitor, it must be that $i\in\Lmsc$, so
\begin{align}\rank(K_{D\perp\Lmsc})&\ge M-\dim\left(\bigcap_{j\in S}\spn(K_{j})\right)\\
&\ge M-\rank(K_{i})\\&\ge M-2\end{align}
where we used the fact that node $i$ has at most two output edges. Since $K_{D\perp\Lmsc}$ has rank at least $M-2$, this is a large enough space for the destination to decode the entire message. The follow Lemma allows us to conclude that all variables in the subspace spanned by $K_{D\perp\Lmsc}$ are trustworthy.
\begin{lemma}\label{lemma:pairwise}
Consider any pair of nodes $i,j$. Suppose $i$ is the traitor, and acts in a way such that $j\in\Lmsc$. Node $i$ cannot have corrupted any value in $K_{D\perp\{i,j\}}\Wbf$.
\end{lemma}

\subsection{Proof of Lemma~\ref{lemma:labels}}\label{subsection:labels}

We begin with $\phi(e)=\psi(e)=\emptyset$ for all edges $e$, and set $\phi$ and $\psi$ progressively. First we describe some properties of the graph $(V,E)$ imposed by the fact that the right hand sides of \eqref{eq:cutset2netA} and \eqref{eq:cutset2netB} are never less than $M-2$.

Given a 2-to-1 node $v$, let $\Gamma_v$ be the set of nodes for which $v$ is the only reachable 2-to-1 node. Note that other than $v$, the only nodes in $\Gamma_v$ are 2-to-2. Moreover, if $v$ can reach another 2-to-1 node, $\Gamma_v$ is empty. We claim that $\Gamma_v$ forms a path. If it did not, then there would be two 2-to-2 nodes $i_1,i_2\in\Gamma_v$ for which there is no path between them. That is, $d_{i_1,i_2}=1$ and $c_{i_1}=c_{i_2}=2$, so \eqref{eq:cutset2netB} becomes $C\le M-3$, which contradicts our assumption that the cut-set bound is $M-2$.

Furthermore, every 2-to-2 node must be able to reach at least one 2-to-1 node. If not, then we could follow a path from such a 2-to-2 node until reaching a node $i_1$ all of whose output edges lead directly to the destination. Node $i_1$ cannot be 2-to-1, so it must be 2-to-2, meaning $e_{i_1}=2$. Taking any other node $i_2$ with a direct link to the destination gives no more than $M-3$ for the right hand side of \eqref{eq:cutset2netA}, again contradicting our assumption.

The first step in the edge labeling procedure is to specify the edges holding non-destination symbols; that is, for each 2-to-1 node $v$, to specify the edges $e$ for which $\phi(e)=v$ and $\psi(e)=1$. To satisfy property (A), these must form a path. For any node $i\in\Nmsc_\inn(D)$, the output edge of $i$ that goes to the destination has no $\phi$ value, so to satisfy property (C), the other output edge $e$ must satisfy $\psi(e)=1$. Moreover, by property (B), if $\phi(e)=v$, then there is a path from $\head(e)$ to $v$. Hence, if $i\in\Vmsc_{2,2}\cap\Gamma_v$ for some 2-to-1 node $v$, then it is impossible for the two output edges of $i$ to have different $\phi$ values; hence, by property (C), one of its output edges $e$ must satisfy $\psi(e)=1$. Therefore, we need to design the non-destination paths so that they pass through $\Gamma_v$ for each $v$, as well as each node in $\Nmsc_\inn(D)$.

For each 2-to-1 node $v$, we first set the end of the non-destination path associated with $v$ to be the edges in $\Gamma_v$. That is, for an edge $e$, if $\head(e),\tail(e)\in\Gamma_v$, set $\psi(e)=1$ and $\phi(e)=v$. Now our only task is to extend the paths backwards such that one is guaranteed to pass through each node in $\Nmsc_\inn(D)$.

Construct an embedding of the graph $(V,E)$ in the plane such that $S$ is on the exterior face. Such an embedding always exists \cite{Harary:72Book}. If we select a set of edges making up an \emph{undirected cycle}---that is, edges constituting a cycle on the underlying undirected graph---then all nodes in the network not on the cycle are divided into those on the interior and those on the exterior, according to the planar embedding. Take $i,j\in\Nmsc_\inn(D)$ such that $i$ can reach $j$, and let $\Cmsc_{i,j}$ be the undirected cycle composed of a path from $i$ to $j$, in addition to the edges $(i,D)$ and $(j,D)$. We claim that if a node $k\in\Nmsc_\inn(D)$ is on the interior of $\Cmsc_{i,j}$, then it is reachable from $i$. Since $S$ is on the exterior face of the graph, it must be exterior to the cycle $\Cmsc_{i,j}$. There exists some path from $S$ to $k$, so it must cross the $\Cmsc_{i,j}$ at a node $j'$. Observe that $j'$ must be on the path from $i$ to $j$, so it is reachable from $i$. Therefore $i$ can reach $j'$ and $j'$ can reach $k$, so $i$ can reach $k$. This construction is diagrammed in Fig.~\ref{fig:cycle1}.

\begin{figure}
\begin{psfrags}
\psfrag{i}[][][.8]{$i$}
\psfrag{j}[][][.8]{$j$}
\psfrag{s}[][][.8]{$S$}
\psfrag{d}[][][.8]{$D$}
\psfrag{k1}[][][.8]{$k$}
\psfrag{jp}[][][.8]{$j'$}
\centerline{\includegraphics[scale=.7]{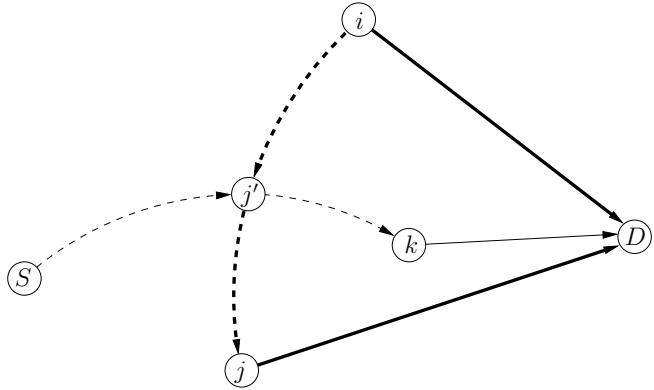}}
\end{psfrags}
\caption[Diagram of planarity being used to prove that a node $k\in\Nmsc_\inn(D)$ on the interior of $\Cmsc_{i,j}$ is reachable from $i$]{A diagram of the planar embedding being used to prove that a node $k\in\Nmsc_\inn(D)$ on the interior of $\Cmsc_{i,j}$ is reachable from $i$. Solid lines are single edges; dashed lines represent paths made up of possibly many edges. Thick lines correspond to edges in $\Cmsc_{i,j}$.}
\label{fig:cycle1}
\end{figure}

We may travel around node $D$ in the planar embedding, noting the order in which the nodes $\Nmsc_\inn(D)$ connect to $D$. Call this order $u_1,\ldots,u_M$. Take any $i\in\Nmsc_\inn(D)$, and suppose $i=u_l$. We claim that the set of nodes in $\Nmsc_\inn(D)$ reachable from $u_l$ forms a contiguous block around $u_l$ in the $\{u\}$ ordering, where we regard $u_1$ and $u_M$ as being adjacent, so two contiguous blocks containing $u_1$ and $u_M$ is considered one contiguous block.

Suppose this were not true. That is, for some $i\in\Nmsc_\inn(D)$ there exists a $j\in\Nmsc_\inn(D)$ reachable from $i$ that is flanked on either side in the $\{u\}$ ordering by nodes $k_1,k_2\in\Nmsc_\inn(D)$ not reachable from $i$. The order in which these four nodes appear in $\{u\}$ in some cyclic permutation or reflection of
\beq(i,k_1,j,k_2).\label{eq:nodeorder}\eeq
Neither $k_1$ nor $k_2$ can be on the interior of $\Cmsc_{i,j}$, because, as shown above, any such node is reachable from $i$. However, if they are both on the exterior, then the order in \eqref{eq:nodeorder} cannot occur, because $D$ is on the boundary of $\Cmsc_{i,j}$.

By contiguity, if a node $i\in\Nmsc_\inn(D)$ can reach any other node in $\Nmsc_\inn(D)$, it can reach a node immediately adjacent to it in the $\{u\}$ ordering. Suppose $i$ can reach both the node $j_1\in\Nmsc_\inn(D)$ immediately to its left and the node $j_2\in\Nmsc_\inn(D)$ immediately to its right. We show that in fact $i$ can reach every node in $\Nmsc_\inn(D)$. In particular, there can be only one such node, or else there would be a cycle. Node $i$ has only two output edges, one of which goes directly to $D$. Let $i'$ be the tail of the other. Both $j_1$ and $j_2$ must be reachable from $i'$.

We claim it is impossible for both $j_1$ to be exterior to $\Cmsc_{i,j_2}$ and $j_2$ to be exterior to $\Cmsc_{i,j_1}$. Suppose both were true. We show the graph must contain a cycle. Let $\bar{\Cmsc}$ be the undirected cycle composed of the path from $i'$ to $j_1$, the path from $i'$ to $j_2$, and the edges $(j_1,D),(j_2,D)$. Every node on $\bar{\Cmsc}$ is reachable from $i$. Since both $j_1$ is exterior to $\Cmsc_{i,j_2}$ and $j_2$ is exterior to $\Cmsc_{i,j_1}$, it is easy to see that $i$ must be on the interior of $\bar{\Cmsc}$. Therefore any path from $S$ to $i$ must cross the cycle at a node $k'$, reachable from $i$. Since $k'$ is on a path from $S$ to $k'$, $i$ is also reachable from $k'$, so there is a cycle. See Fig.~\ref{fig:cycle2} for a diagram of this.

\begin{figure}
\begin{psfrags}
\psfrag{i}[][][.8]{$i$}
\psfrag{ip}[][][.8]{$i'$}
\psfrag{j1}[][][.8]{$j_1$}
\psfrag{s}[][][.8]{$S$}
\psfrag{d}[][][.8]{$D$}
\psfrag{j2}[][][.8]{$j_2$}
\psfrag{kp}[][][.8]{$k'$}
\psfrag{cij1}[][][.8]{$\Cmsc_{i,j_1}$}
\psfrag{cij2}[][][.8]{$\Cmsc_{i,j_2}$}
\centerline{\includegraphics[scale=.7]{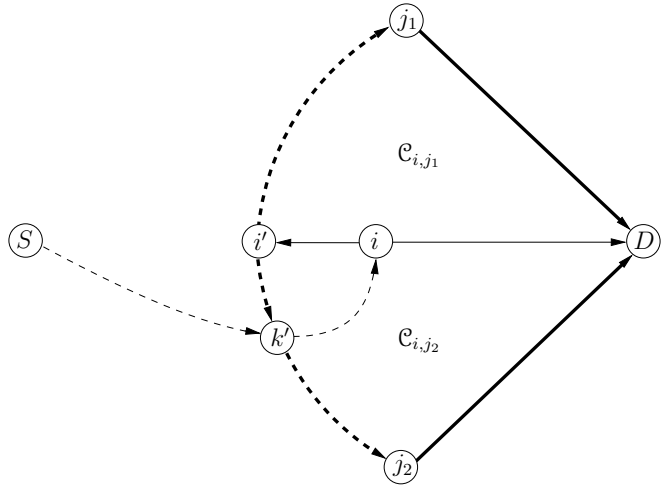}}
\end{psfrags}
\caption[Diagram of planarity being used to prove that a node reaching its two neighbors in $\Nmsc_\inn(D)$ can reach every node in $\Nmsc_\inn(D)$]{A diagram of the planar embedding being used to prove that a node reaching its two neighbors in $\Nmsc_\inn(D)$ can reach every node in $\Nmsc_\inn(D)$. Solid lines are single edges; dashed lines represent paths made up of possibly many edges. Thick lines correspond to the undirected cycle $\bar{\Cmsc}$. Undirected cycles $\Cmsc_{i,j_1}$ and $\Cmsc_{i,j_2}$ are indicated.}
\label{fig:cycle2}
\end{figure}

Therefore, we may assume without loss of generality that $j_2$ is in the interior of $\Cmsc_{i,j_1}$. Suppose there were a node $j_3\in\Nmsc_\inn(D)$ not reachable from $i$. Node $j_3$ must be on the exterior of $\Cmsc_{i,j_1}$, because we have shown that nodes in $\Nmsc_\inn(D)$ on the interior are reachable from $i$. Therefore, in the $\{u\}$ order, these four nodes must appear in some cyclic permutation or reflection of $(i,j_3,j_1,j_2)$. However, this is impossible, because both $j_1$ and $j_2$ were assumed to be adjacent to $i$. Therefore, $i$ can reach every node in $\Nmsc_\inn(D)$.

Take a node $i$ that can reach 2-to-1 nodes $v_1,v_2\in\Nmsc_\inn(D)$. Suppose that $i$ cannot reach every node in $\Nmsc_\inn(D)$. Therefore, the nodes it can reach in in $\Nmsc_\inn(D)$ are either entirely to its right or entirely to its left in the $\{u\}$ ordering, or else, by contiguity, node $i$ would be able to reach the adjacent nodes on both sides. Suppose without loss of generality that they are all to its right, and that $v_2$ is further to the right than $v_1$. We claim that $v_1$ is on the interior of $\Cmsc_{i,v_2}$. Suppose it were on the exterior. By contiguity, every node in $\Nmsc_\inn(D)$ on the exterior of $\Cmsc_{i,v_2}$ must be reachable from $i$. Since we have already argued that every node in $\Nmsc_\inn(D)$ on the interior of $\Cmsc_{i,v_2}$ is reachable from $i$, this means $i$ can reach every node in $\Nmsc_\inn(D)$, which we have assumed is not the case.

Therefore, $v_1$ is on the interior of $\Cmsc_{i,v_2}$. We may construct a path from $S$ to $v_1$, passing through all nodes in $\Gamma_{v_1}$. This path must cross $\Cmsc_{i,v_2}$ at a node $k$, reachable from $i$. Node $j$ can reach both $v_1$ and $v_2$, so it cannot be in $\Gamma_{v_1}$. However, $j$ is on a path passing through $\Gamma_{v_1}$, so it can reach all nodes in $\Gamma_{v_1}$. Therefore there exists a path from $i$ to $v_1$, passing through $\Gamma_{v_1}$.

If $i$ can reach every node in $\Nmsc_\inn(D)$, then as shown above, either $v_1$ is in the interior of $\Cmsc_{i,v_1}$, or $v_2$ is in the interior of $\Cmsc_{i,v_2}$. Therefore, by the same argument to that just used for the case that $i$ cannot reach every node in $\Nmsc_\inn(D)$, there is either a path from $i$ to $v_1$ through $\Gamma_{v_1}$ or a path from $i$ to $v_2$ through $\Gamma_{v_2}$.

Fix a 2-to-1 node $v_1\in\Nmsc_\inn(D)$. Consider the set of nodes that are:
\begin{itemize}
\item contained in $\Vmsc_{2,2}\cap\Nmsc_\inn(D)$,
\item not in $\Gamma_v$ for any 2-to-1 node $v$,
\item can reach $v_1$,
\item cannot reach any other node also satisfying the above three conditions.
\end{itemize}
We claim there are at most two such nodes. Suppose there were two such nodes $i_1,i_2$ both to the left of $v_1$ in the $\{u\}$ ordering. If $i_1$ were further to the left, then $i_1$ could reach $i_2$, since $i_1$ can reach $v_1$ and the nodes reachable from $i_1$ must form a contiguous block. Hence $i_1$ would not qualify. Therefore there can be at most one such node to the left of $v_1$ and at most one to the right. Denote these two nodes $i$ and $j$ respectively, if they exist. By contiguity, every node satisfying the first three conditions must be able to reach either $i$ or $j$. Moreover, all such nodes to the left of $v_1$ form a single path ending in $i$, and those on the right form a single path ending in $j$. We will proceed to extend two non-destination paths backwards to $i$ and $j$. Then, we may further extend these two paths backwards through all nodes in $\Vmsc_{2,2}\cap\Nmsc_\inn(D)$ that can reach $v_1$, and then backwards to the source on arbitrary paths. Hence, we need only find paths from $i$ to the head of $\Gamma_v$ for some $v$, and a distinct one of the same for $j$.

Both $i$ and $j$ can reach at least one 2-to-1 node other than $v_1$. Suppose $i$ can reach another 2-to-1 node $v_2\in\Nmsc_\inn(D)$. By the argument above, there is a path from $i$ to the leftmost of $v_1,v_2$ through $\Gamma_{v_1}$ or $\Gamma_{v_2}$ respectively. Similarly, if $j$ can reach a 2-to-1 node $v_3\in\Nmsc_\inn(D)$ with $v_3\ne v_1$, there is a path from $j$ to the rightmost of $v_1,v_3$, through the associated $\Gamma$. This is true even if $v_2=v_3$.

Suppose there is no 2-to-1 node in $\Nmsc_\inn(D)$ reachable from node $i$ other than $v_1$. There still must be a 2-to-1 node $v_2$ reachable from $i$, though $v_2\notin\Nmsc_\inn(D)$. Since $v_2$ is not adjacent to the destination, it must be able to reach a 2-to-1 node that is. Therefore $\Gamma_{v_2}=\emptyset$, so any path from $i$ to $v_2$ trivially includes $\Gamma_{v_2}$. If $j$ can also reach no 2-to-1 nodes in $\Nmsc_\inn(D)$ other than $v_1$, there must be some 2-to-1 node $v_3\notin\Nmsc_\inn(D)$ reachable from $j$. We may therefore select non-destination paths from $i$ to $v_2$ and $j$ to $v_3$, unless $v_2=v_3$. This only occurs if this single node is the only 2-to-1 node other than $v_1$ reachable by either $i$ or $j$. We claim that in this case, either $i$ or $j$ can reach the tail of $\Gamma_{v_1}$. Therefore we may extend the non-destination path for $v_1$ back to one of $i$ or $j$, and the non-destination path for $v_2=v_3$ to the other. Every node can reach some 2-to-1 node in $\Nmsc_\inn(D)$, so $v_2$ can reach $v_1$, or else $i$ and $j$ would be able to reach a different 2-to-1 node in $\Nmsc_\inn(D)$. By a similar argument to that used above, $v_1$ must be on the interior of the undirected cycle composed of the path from $i$ to $v_2$, the path from $j$ to $v_2$, and the edges $(i,D),(j,D)$. If not, $v_1$ would not be between $i$ and $j$ in the $\{u\}$ ordering. Note this is true even if $i$ can reach $j$ or vice versa. Since $S$ must be exterior to this cycle, any path from $S$ to $v_1$ including $\Gamma_{v_1}$ must cross either the path from $i$ to $v_2$ or $j$ to $v_2$ at a node $k$. Node $k$ must be able to reach the head of $\Gamma_{v_1}$, so either $i$ or $j$ can reach $\Gamma_{v_1}$.

Once the non-destination paths are defined, we perform the following algorithm to label other edges so as to satisfy property (C). We refer to an edge $e$ as \emph{labeled} if $\phi(e)\ne\emptyset$. We refer to a node as \emph{labeled} if any of its output edges are labeled. Any node unlabeled after the specifications of the non-destination paths must not be in $\Nmsc_\inn(D)$, and must be able to reach at least two different 2-to-1 nodes.

\begin{enumerate}
\item For any edge $e$ such that there exists an $e'\in\Emsc_\out(\tail(e))$ with $\psi(e')=1$, set $\phi(e)=\phi(e')$. Observe now that any path eventually reaches a labeled edge. Furthermore, the tail of any unlabeled edge cannot be a node contained in $\Gamma_v$ for any $v$, so it can lead to at least two 2-to-1 nodes.
\item Repeat the following until every edge other than those connected directly to the destination is labeled. Consider two cases:
    \begin{itemize}
    \item\emph{There is no 2-to-2 node with exactly one labeled output edge}: Pick an unlabeled node $i$. Select any path of unlabeled edges out of $i$ until reaching a labeled node. Let $v$ be the label of a labeled output edge from this node. For all edges $e$ on the selected path, set $\phi(e)=v$. Observe that every node on this path was previously an unlabeled 2-to-2 node. Hence every node on this path, except the last one, has exactly one labeled output edge.
    \item\emph{There is a 2-to-2 node $i$ with exactly one labeled output edge}: Let $v_1$ be the label on the labeled output edge. Select any path of unlabeled edges beginning with the unlabeled output edge from $i$ until reaching a node with an output edge labeled $v_2$ with $v_2\ne v_1$. This is always possible because any unlabeled edge must be able to lead to at least two 2-to-1 nodes, including one other than $v_1$. For all edges $e$ on the selected path, set $\phi(e)=v_2$. Observe that before we labeled the path, no node in the path other than the last one had an output edge labeled $v_2$, because if it did, we would have stopped there. Hence, after we label the path, if a node now has 2 labeled output edges, they have different labels.
    \end{itemize}
\end{enumerate}
Note that in the above algorithm, whenever an edge $e$ becomes labeled, if there was another edge $e'$ with $\head(e)=\head(e')$, either $e'$ was unlabeled, or $\phi(e)\ne\phi(e')$. Therefore, the final $\phi$ values satisfy property (B).

\subsection{Proof of Lemma~\ref{lemma:tree}}

Observe that for any 2-to-2 node, the two $\rho$ values on the input edges are identical to the two $\rho$ values on the output edges. For a 2-to-1 node, the $\rho$ value on the output edge is equal to the $\rho$ value on one of the input edges. Therefore beginning with any edge in $\Emsc_\out(S)$, we may follow a path along only edges with the same $\rho$ value, and clearly we will hit all such edges. Property (1) immediately follows.

Property (2) follows from the fact that 2-to-2 nodes always operate such that from the symbols on the two output edges, it is possible to decode the symbols on the input edges. Therefore the destination can always reverse these transformations to recover any earlier symbols sent in the network. The only exception is 2-to-1 nodes, which drop one of their two input symbols. The dropped symbol is a non-destination symbol, so it is clear that the destination can always decode the rest.

We now prove property (3). We claim that when the comparison fails at node $k$, it is impossible for the destination to decode $X_{f_2}$. We may assume that the destination has direct access to all symbols on edges immediately subsequent to edges in $\Lambda(v)$. This can only make $X_{f_2}$ easier to decode. Recall that $\rho(f_1)<\rho(f_2)$, so $X_{f_1}$ is forwarded directly on the output edge of $k$ not in $\Lambda(v)$. Therefore the destination can only decode $X_{f_2}$ if it can decode the symbol on the output edge of $k$ in $\Lambda(v)$. Continuing to follow the path through $\Lambda(v)$, suppose we reach an edge $e_1$ with $\tail(e_1)=k'$, where $k'$ is a branch node. Let $e_2$ be the other input edge of $k'$. Even if $\rho(e_1)<\rho(e_2)$, meaning $k'$ would normally forward $X_{e_1}$ outside of $\Lambda(v)$, because $e_1$ carries a failed comparison bit, $k'$ will instead forward $X_{e_2}$ outside of $\Lambda(v)$. Again, the destination can only decode $X_{f_2}$ (or equivalently $X_{e_1}$) if it can decode the symbol on the output edge of $k'$ in $\Lambda(v)$. If we reach a node interacting with the non-destination symbol associated with $v$, then because of the failed comparison bit, the formerly non-destination symbol is forwarded outside of $\Lambda(v)$ and the symbol to decode continues traveling through $\Lambda(v)$. It will finally reach $v$, at which point it is dropped. Therefore it is never forwarded out of $\Lambda(v)$, so the destination cannot recover it.

\subsection{Proof of Lemma~\ref{lemma:vander}}

From Corollary~\ref{cor:magic3}, it is enough to prove the existence of a $K$ matrix satisfying
\beq|K_{e_{v_1}^*,e_{v_2}^*,\Zbf}|\ |K_{f_1,f_2,\Zbf}|\ |K_{e_{v_1}^*,f_1,\Zbf}|\ |K_{e_{v_2}^*,f_2,\Zbf}|<0.\label{eq:detvand}\eeq
We construct a Vandermonde matrix $K$ to satisfy \eqref{eq:detvand} for all $v_1,v_2$ and all $f_1,f_2$ in the following way. We will construct a bijective function (an ordering) $\alpha$ given by
\beq\alpha:\{e_v^*:v\in\Vmsc_{2,1}\}\cup\Nmsc_{\text{in}}(D)\to\{1,\ldots,M+|\Vmsc_{2,1}|\}.\eeq
For each $v\in\Vmsc_{2,1}$, set $\alpha(e_v^*)$ to an arbitrary but unique number in $1,\ldots,|\Vmsc_{2,1}|$. We may now refer to a 2-to-1 node as $\alpha^{-1}(a)$ for an integer $a\in\{1,\ldots,|\Vmsc_{2,1}|\}$. Now set $\alpha(e)$ for $e\in\Emsc_\inn(D)$ such that, in $\alpha$ order, the edge set $\{e_v^*:v\in\Vmsc_{2,1}\}\cup\Nmsc_{\text{in}}(D)$ is written
\begin{multline} e_{\alpha^{-1}(1)}^*,e_{\alpha^{-1}(2)}^*,\ldots,e_{\alpha^{-1}(|\Vmsc_{2,1}|)}^*,\\
\Xi(\alpha^{-1}(|\Vmsc_{2,1}|)),\Xi(\alpha^{-1}(|\Vmsc_{2,1}|-1)),\ldots,\Xi(\alpha^{-1}(1)).\end{multline}
That is, each $\Xi(v)$ set is consecutive in the ordering, but in the opposite order as the associated non-destination edges $e_v^*$. Now let $K$ be the Vandermone matrix with constants given by $\alpha$. That is, the row associated with edge $e$ is given by
\beq\left[\begin{array}{ccccc}1&\alpha(e)&\alpha(e)^2&\cdots&\alpha(e)^{M-3}\end{array}\right].\label{eq:vandrow}\eeq

We claim the matrix $K$ given by \eqref{eq:vandrow} satisfies \eqref{eq:detvand}. Fix $v_1,v_2$, and $f_1\in\Xi(v_1),f_2\in\Xi(v_2)$. Due to the Vandermonde structure of $K$, we can write the determinant of a square submatrix in terms of the constants $\alpha(e)$. For instance,
\begin{multline}
|K_{e_{v_1}^*,e_{v_2}^*,\Zbf}|=[\alpha(e_{v_2}^*)-\alpha(e_{v_1}^*)]
\prod_{e\in \Zbf}[\alpha(e)-\alpha(e_{v_1}^*)][\alpha(e)-\alpha(e_{v_2}^*)]\\
\cdot\prod_{e,e'\in \Zbf,\alpha(e)<\alpha(e')}[\alpha(e')-\alpha(e)]\end{multline}
where we have assumed without loss of generality that the rows of $K_\Zbf$ are ordered according to $\alpha$. Expanding the determinants in \eqref{eq:detvand} as such gives
\begin{align}
&|K_{e_{v_1}^*,e_{v_2}^*,\Zbf}|\ |K_{f_1,f_2,\Zbf}|\ |K_{e_{v_1}^*,f_1,\Zbf}|\ |K_{e_{v_2}^*,f_2,\Zbf}|\\
&=[\alpha(e_{v_2}^*)-\alpha(e_{v_1}^*)][\alpha(f_2)-\alpha(f_1)]
[\alpha(f_1)-\alpha(e_{v_1}^*)][\alpha(f_2)-\alpha(e_{v_2}^*)]\nonumber\\
&\ \ \cdot\prod_{e\in \Zbf}[\alpha(e)-\alpha(e_{v_1}^*)]^2[\alpha(e)-\alpha(e_{v_2}^*)]^2
[\alpha(e)-\alpha(f_1)]^2[\alpha(e)-\alpha(f_1)]^2\nonumber\\
&\ \ \cdot\prod_{e,e'\in \Zbf,\alpha(e)<\alpha(e')}[\alpha(e')-\alpha(e)]^4.\label{eq:detexpand}\end{align}
Recall $f_1\in\Xi(v_1),f_2\in\Xi(v_2)$. Since we chose $\alpha$ such that the $\Xi$ sets are in opposite order to the edges $e_v^*$, we have
\beq [\alpha(e_{v_2}^*)-\alpha(e_{v_1}^*)][\alpha(f_2)-\alpha(f_1)]<0.\eeq
Moreover, since all the $\Xi$ sets have larger $\alpha$ values than the edges $e_v^*$,
\begin{align}
\alpha(f_1)-\alpha(e_{v_1}^*)>0,\\
\alpha(f_2)-\alpha(e_{v_2}^*)>0.
\end{align}
Hence, there is exactly one negative term in \eqref{eq:detexpand}, from which we may conclude \eqref{eq:detvand}.

\subsection{Proof of Lemma~\ref{lemma:pairwise}}\label{subsection:pairwise}

The random vector $\Wbf$ is distributed according to the type of the message vector as it is produced as the source. We formally introduce the random vector $\tWbf$ representing the message as it is transformed in the network. As in our examples, this vector is distributed according to the joint type of the sequences as they appear in the network, after being corrupted by the adversary. For each edge $e$, we define $X_e$ and $\tX_e$ similarly as random variables jointly distributed with $\Wbf$ and $\tWbf$ respectively with distributions given by the expected and corrupted joint types.

For every pair of nodes $(i,j)$, we need to prove both of the following:
\begin{align}
&\text{\parbox{4.8in}{If $i$ is the traitor, and $j\in\Lmsc$, $i$ cannot corrupt values in $K_{D\perp\{i,j\}}\Wbf$.}}\label{eq:itoj}\\
&\text{\parbox{4.8in}{If $j$ is the traitor, and $i\in\Lmsc$, $j$ cannot corrupt values in $K_{D\perp\{i,j\}}\Wbf$.}}\label{eq:jtoi}
\end{align}
In fact, each of these implies the other, so it will be enough to prove just one. Suppose \eqref{eq:itoj} holds. Therefore, if the distribution observed by the destination of $K_{D\perp\{i,j\}}\tWbf$ does not match that of $K_{D\perp\{i,j\}}\tWbf$, then at least one of $i,j$ will not be in $\Lmsc$. If they both were in $\Lmsc$, it would have had to be possible for node $i$ to be the traitor, make it appear as if node $j$ were the traitor, but also corrupt part of $K_{D\perp\{i,j\}}W$. By \eqref{eq:itoj}, this is impossible. Hence, if $j$ is the traitor and $i\in\Lmsc$, then the distribution of the $K_{D\perp\{i,j\}}Y_D$ must remain uncorrupted. This vector includes $K_{D\perp j}W$, a vector that can certainly not be corrupted by node $j$. Since $\rank(K_{D\perp j})\ge M-2$, and there are only $M-2$ degrees of freedom, the only choice node $j$ has to ensure that the distribution of $K_{D\perp\{i,j\}}W$ matches $p$ is to leave this entire vector uncorrupted. That is, \eqref{eq:jtoi} holds.

Fix a pair $(i,j)$. We proceed to prove either \eqref{eq:itoj} or \eqref{eq:jtoi}. Doing so will require placing constraints on the actions of the traitor imposed by comparisons that occur inside the network, then applying one of the corollaries of Theorem~\ref{thm:magic} in Sec.~\ref{sec:polytope}.  Let $K_{\perp i}$ be a basis for the space orthogonal to $K_i$. If node $i$ is the traitor, we have that $K_{\perp i}\tWbf\sim K_{\perp i}\Wbf$. Moreover, since $j\in\Lmsc$, $K_{D\perp j}\tWbf)\sim K_{D\perp j}\Wbf$. These two constraints are analogous to \eqref{eq:distcon2} and \eqref{eq:distcon1} respectively, where the symbols on the output of node $i$ are analogous to $X_1,X_2$. The subspace of $K_D$ orthogonal to both $K_i$ and $K_j$ corresponds to $\Zbf$ in the example. We now seek pairwise constraints of the form \eqref{eq:distcon3}--\eqref{eq:distcon5} from successful comparisons to apply Theorem~\ref{thm:magic}.

Being able to apply Theorem~\ref{thm:magic} requires that $K_{D\perp j}$ has rank $M-2$ for all $j$. Ensuring this has to do with the choices for the coefficients $\gamma_{i,1},\gamma_{i,2}$ used in \eqref{eq:xel}. A rank deficiency in $K_{D\perp j}$ is a singular event, so it is not hard to see that random choices for the $\gamma$ will cause this to occur with small probability. Therefore such $\gamma$ exist.

We now discuss how pairwise constraints on the output symbols of $i$ or $j$ are found. Consider the following cases and subcases:
\begin{itemize}
\item\emph{$i,j\in\Wmsc_1\cup\Wmsc_2$}:
    Suppose node $i$ is the traitor. Let $e_1,e_2$ be the output edges of node $i$. For each $l=1,2$, we look for constraints on $X_{e_l}$ by following the $\rho=\rho(e_l)$ path until one of the following occurs:
    \begin{itemize}
    \item\emph{We reach an edge on the $\rho=\rho(e_l)$ path carrying a symbol influenced by node $j$}: This can only occur immediately after a branch node $k$ with input edges $f_1,f_2$ where $\rho(f_1)=\rho(e_l)$, $\rho(f_2)<\rho(f_1)$, and $X_{f_2}$ is influenced by node $j$. At node $k$, a comparison occurs between $\tX_{f_1}$, which is influenced by node $i$ but not $j$, and $\tX_{f_2}$. If the comparison succeeds, then this places a constraint on the distribution of $(\tX_{f_1},\tX_{f_2})$. If the comparison fails, the forwarding pattern changes such that the $\rho=\rho(e_l)$ path becomes a non-destination path; i.e. the value placed on $e_l$ does not affect any variables available at the destination. Hence, the subspace available at the destination that is corruptible by node $i$ is of dimension at most one.
    \item\emph{We reach node $j$ itself}: In this situation, we make use of the fact that we only need to prove that node $i$ cannot corrupt values available at the destination that cannot also be influenced by node $j$. Consider whether the $\rho=\rho(e_l)$ path, between $i$ and $j$, contains a branch node $k$ with input edges $f_1,f_2$ such that $\rho(f_1)=\rho(e_l)$ and $\rho(f_2)>\rho(f_1)$. If there is no such node, then $X_{e_l}$ cannot influence any symbols seen by the destination that are not also being influenced by $j$. That is, $X_{e_l}$ is in $\spn(K_{i\to D}\cap K_{j\to D})$, so we do not have anything to prove. If there is such a branch node $k$, then the output edge $e$ of $k$ with $\rho(e)=\rho(f_2)$ contains a symbol influenced by $i$ and not $j$. We may now follow the $\rho=\rho(e)$ path from here to find a constraint on $X_{e_l}$. If a comparison fails further along causing the forwarding pattern to change such that the $\rho=\rho(e)$ path does not reach the destination, then the potential influence of $X_{e_l}$ on a symbol seen by the destination not influenced by node $j$ is removed, so again we do not have anything to prove.
    \item\emph{The $\rho=\rho(e_l)$ path leaves the network without either of the above occurring}: Immediately before leaving the network, the symbol will be compared with a non-destination symbol. This comparison must succeed, because $j$ cannot influence the non-destination symbol. This gives a constraint $\tX_{e_l}$.
    \end{itemize}
    We may classify the fates of the two symbols out of $i$ as discussed above as follows:
    \begin{enumerate}
    \item Either the forwarding pattern changes such that the symbol does not reach the destination, or the symbol is in $\spn(K_{i\to D}\cap K_{j\to D})$, and so we do not need to prove that it cannot be corrupted. Either way, we may ignore this symbol.
    \item The symbol leaves the network, immediately after a successfully comparison with a non-destination symbol.
    \item The symbol is successfully compared with a symbol influenced by node $j$. In particular, this symbol from node $j$ has a strictly smaller $\rho$ value than $\rho(e_l)$.
    \end{enumerate}
    We divide the situation based on which of the above cases occur for $l=1,2$ as follows:
    \begin{itemize}
    \item\emph{Case 1 occurs for both $l=1,2$}: We have nothing to prove.
    \item\emph{Case 1 occurs for (without loss of generality) $l=1$}: Either case 2 or 3 gives a successful comparison involving a symbol influenced by $\tX_{e_{\bar{l}}}$. Applying Corollary~\ref{cor:magic1} allows us to conclude that $\tX_{e_{\bar{l}}}$ cannot be corrupted.
    \item\emph{Case 2 occurs for both $l=1,2$}: If the two paths reach different non-destination symbols, then we may apply Lemma~\ref{lemma:vander} to conclude that node $i$ cannot corrupt either $\tX_{e_1}$ nor $\tX_{e_2}$. Suppose, on the other hand, that each path reaches the same non-destination path, in particular the one associated with 2-to-1 node $v$. Since $\phi(e_1)\ne\phi(e_2)$, assume without loss of generality that $\phi(e_1)\ne v$. We may follow the path starting from $e_1$ through $\Gamma(v)$ to find an additional constraint, after which we may apply Corollary~\ref{cor:magic2}. All symbols on this path are influenced by $\tX_{e_1}$. This path eventually crosses the non-destination path associated with $v$. If the symbol compared against the non-destination symbol at this point is not influenced by $j$, then the comparison succeeds, giving an additional constraint. Otherwise, there are two possibilities:
        \begin{itemize}
        \item\emph{The path through $\Gamma(v)$ reaches $j$}: There must be a branch node on the path to $\Gamma(v)$ before reaching $j$ such that the path from $e_1$ has the smaller $\rho$ value. If there were not, then case 1 would have occurred. Consider the most recent such branch node $k$ in $\Gamma(v)$ before reaching $j$. Let $f_1,f_2$ be the input edges to $k$, where $f_1$ is on the path from $e_1$. We know $\rho(f_1)<\rho(f_2)$. The comparison at $k$ must succeed. Moreover, this successful comparison comprises a substantial constraint, because the only way the destination can decode $X_{f_2}$ is through symbols influenced by node $j$.
        \item\emph{The path through $\Gamma(v)$ does not reach $j$}: Let $k$ be the first common node on the paths from $i$ and $j$ through $\Gamma(v)$. Let $f_1,f_2$ be the input edges of $k$, where $f_1$ is on the path from $i$ and $f_2$ is on the path from $j$. If the comparison at $k$ succeeds, this provides a constraint. If it fails, then the forwarding pattern changes such that the $\rho=\rho(f_1)$ path becomes a non-destination path. Since we are not in case 1, $\rho(e_1)\ne\rho(f_1)$, but a symbol influenced by $X_{e_1}$ is compared against a symbol on the $\rho=\rho(f_1)$ path at a branch node in $\Gamma(v)$. This comparison must succeed, providing an additional constraint.
        \end{itemize}
    \item\emph{Case 3 occurs for (without loss of generality) $l=1$, and either case 2 or 3 occurs for $l=2$}: We now suppose instead that node $j$ is the traitor. That is, we will prove \eqref{eq:jtoi} instead of \eqref{eq:itoj}. Recall that a successful comparison occurs at a branch node $k$ with input edges $f_1,f_2$ where $\tX_{f_1}$ is influenced by $\tX_{e_1}$, $\tX_{f_2}$ is influenced by node $j$, and $\rho(f_2)<\rho(f_1)$. Let $e'_1,e'_2$ be the output edges of node $j$, and suppose that $\rho(e'_1)=\rho(f_2)$; i.e. the symbol $X_{f_2}$ is influenced by $X_{e'_1}$. The success of the comparison gives a constraint on $\tX_{e'_1}$. Since $\rho(f_2)<\rho(f_1)$, we may continue to follow the $\rho=\rho(f_2)$ path from node $k$, and it continues to be not influenced by node $i$. As above, we may find an additional constraint on $X_{e'_1}$ by following this $\rho$ path until reaching a non-destination symbol or reaching another significant branch node. Furthermore, we may find a constraint on $\tX_{e'_2}$ in a similar fashion. This gives three constraints on $\tX_{e'_1},\tX_{e'_2}$, enough to apply Corollary~\ref{cor:magic2}, and conclude that node $j$ cannot corrupt its output symbols.
    \end{itemize}
\item $i\in\Wmsc_3\cup\Vmsc_{2,1}\setminus\Nmsc_\inn(D),j\in\Wmsc_1\cup\Wmsc_2$: Assume node $i$ is the traitor. If $i\in\Vmsc_{2,1}$ with single output edge $e$ such that $\psi(e)=1$, then node $i$ controls no symbols received at the destination and we have nothing to prove. Otherwise, it controls just one symbol received at the destination, so any single constraint on node $i$ is enough. Let $e'$ be the output symbol of $i$ with $\psi(i)=0$. Since we assume $i\notin\Nmsc_\inn(D)$, the $\rho=\rho(e')$ path is guaranteed to cross a non-destination path after node $i$. As above, follow the $\rho=\rho(e')$ path until reaching a branch node $k$ at which the symbol is combined with one influenced by node $j$. If the comparison at node $k$ succeeds, it gives a constraint on $\tX_{e'}$. If the comparison fails, then the forwarding pattern will change such that the $\rho=\rho(e')$ path will fail to reach the destination, so we're done.
\item $i\in\Wmsc_1\cup\Wmsc_2,j\in\Nmsc_\inn(D)$: Assume node $i$ is the traitor. By construction, since one output edge of $j$ goes directly into the destination, the other must be on a non-destination path. Hence, $j$ only controls one symbol at the destination, so we again need to place only one constraint on node $i$. Let $e\in\Emsc_\out(i)$ be such that $\phi(e)\ne\phi(e')$ for all $e'\in\Emsc_\out(j)$. This is always possible, since the two output edges of $i$ have different $\phi$ values, and since one output edge of $j$ goes directly to the destination, only one of the output edges of $j$ has a $\phi$ value. Let $v=\phi(e)$. Follow the path from $e$ through $\Lambda(v)$ until reaching the non-destination symbol at node $k$ with input edges $f_1,f_2$. Assume $\tX_{f_1}$ is influenced by $\tX_e$ and $\tX_{f_2}$ is a non-destination symbol. The comparison between these two symbols must succeed, because node $j$ cannot influence either $\tX_{f_1}$ or $\tX_{f_2}$. This places the necessary constraint on $\tX_e$.
\item\emph{$i,j\in\Wmsc_3\cup\Vmsc_{2,1}$}: Nodes $i,j$ each control at most one symbol available at the destination, so either one, in order to make it appear as if the other could be the traitor, cannot corrupt anything.
\end{itemize}

\subsection{Proof of Theorem~\ref{thm:planar} when the Cut-set Bound is $M-3$}\label{subsection:m3proof}

We now briefly sketch the proof of Theorem~\ref{thm:planar} for the case that the cut-set bound is $M-3$. The proof is far less complicated than the above proof for the $M-2$ case, but it makes use of many of the same ingredients. First note that the set of 2-to-2 nodes $i$ that cannot reach any 2-to-1 nodes must form a path. We next perform a similar edge labeling as above, defining $\phi$ and $\psi$ as in \eqref{eq:phidef}--\eqref{eq:psidef}. Properties (A) and (B) must still hold, except that edges may have null labels, and property (C) is replaced by
\begin{description}
\item[\textbf{C'}] For every 2-to-2 node that can reach at least one 2-to-1 node, at least one of its output edges must have a non-null label.
\end{description}
Internal nodes operate in the same way based on the edge labels as above, where symbols are always forwarded along edges with null labels. The decoding process is the same. Proving an analogous version of Lemma~\ref{lemma:pairwise} requires only finding a single constraint on one of $i$ or $j$. This is always possible since one is guaranteed to have a label on an output edge, unless they are both in the single path with no reachable 2-to-1 nodes, in which case they influence the same symbol reaching the destination.

Interestingly, this proof does not make use of the planarity of the graph. We may therefore conclude that for networks satisfying properties (2) and (3) in the statement of Theorem~\ref{thm:planar}, the cut-set bound is always achievable if the cut-set is strictly less than $M-2$.

\section{Looseness of the Cut-set Bound}\label{sec:loose}

So far, the only available upper bound on achievable rates has been the cut-set bound. We have conjectured that for planar graphs this bound is tight, but that still leaves open the question of whether there is a tighter upper bound for non-planar graphs. It was conjectured in \cite{KimEtal:09Allerton} that there is such a tighter bound, and here we prove this conjecture to be true. We do this in two parts. First, consider the problem that the traitor nodes in a Byzantine attack are constrained to be only from a certain subset of nodes. That is, a special subset of nodes are designated as potential traitors, and the code must guard against adversarial control of any $z$ of those nodes. We refer to this as the limited-node problem. Certainly the limited-node problem subsumes the all-node problem, since we may simply take the set of potential traitors to be all nodes. Furthermore, it subsumes the unequal-edges problem studied in \cite{KimEtal:09Allerton}, because given an instance of the unequal-edge problem, an equivalent all-node problem can be constructed as follows: create a new network with every edge replaced by a pair of edges of equal capacity with a node between them. Then limit the traitors to be only these interior nodes.

We will show in Section~\ref{subsec:limited} that the all-node problem actually subsumes the limited-node problem, and therefore also the unequal-edge problem. Then in Section \ref{subsec:boundexample} we give an example of a limited-node network for which there is an active upper bound on capacity other than the cut-set. This proves that, even for the all-node problem, the cut-set bound is not tight in general. Transforming the example in Section~\ref{subsec:boundexample} into an unequal-edge problem is not hard; this therefore confirms the conjecture in \cite{KimEtal:09Allerton}.

\subsection{Equivalence of Limited-Node and All-Node}\label{subsec:limited}

Let $(V,E)$ be a network under a limited-node Byzantine attack, where there may be at most $z$ traitors constrained to be in $U\subseteq V$, and let $C$ be its capacity. We construct a sequence of all-node problems, such that finding the capacity of these problems is enough to find that of the original limited-node problem. Let $(V^{(M)},E^{(M)})$ be a network as follows. First make $M$ copies of $(V,E)$. That is, for each $i\in V$, put $i^{(1)},\ldots,i^{(M)}$ into $V^{(M)}$, and for each edge $(i,j)\in E$, put $(i^{(1)},j^{(1)},\ldots,(i^{(M)},j^{(M)})$ into $E^{(M)}$. Then, for each $i\in U$, merge $i^{(1)},\ldots,i^{(M)}$ into a single node $i^*$, transferring all edges that were previously connected to any of $i^{(1)},\ldots,i^{(M)}$ to $i^*$. Let $C^{(M)}$ be the all-node capacity of $(V^{(M)},E^{(M)})$ with $z$ traitors. This construction is illustrated in Fig.~\ref{fig:42net}, where we show a limited-node network $(V,E)$ and the all-node network $(V^{(M)},E^{(M)})$ with $M=3$. For large $M$, the all-node problem will be such that for any $i\notin U$, the adversary has no reason to control one of the respective nodes because it commands such a small fraction of the information flow through the network. That is, we may assume that the traitors will only ever be nodes in $U$. This is stated explicitly in the following theorem.
\begin{theorem}\label{thm:limited}
For any $M$, $C^{(M)}$ is related to $C$ by
\beq \frac{1}{M}C^{(M)}\le C\le \frac{1}{M-2z}C^{(M)}.\label{eq:limitedbounds}\eeq
Moreover,
\beq C=\lim_{M\to\infty}\frac{1}{M}C^{(M)}\label{eq:limitedlimit}\eeq
and if $C^{(M)}$ can be computed to arbitrary precision for any $M$ in finite time, then so can $C$.
\end{theorem}
\begin{proof}
We first show that $\frac{1}{M}C^{(M)}\le C$. Take any code on $(V^{(M)},E^{(M)})$ achieving rate $R$ when any $z$ nodes may be traitors. We use this to construct a code on $(V,E)$, achieving rate $R/M$ when any $z$ nodes in $U$ may be traitors. We do this by first increasing the block-length by a factor of $M$, but maintaining the same number of messages, thereby reducing the achieved rate by a factor of $M$. Now, since each edge in $(V,E)$ corresponds to $M$ edges in $(V^{(M)},E^{(M)})$, we may place every value transmitted on an edge in the $(V^{(M)},E^{(M)})$ code to be transmitted on the equivalent edge in the $(V,E)$ code. That is, all functions executed by $i^{(1)},\ldots,i^{(M)}$ are now executed by $i$. The original code could certainly handle any $z$ traitor nodes in $U$. Hence the new code can handle any $z$ nodes in $U$, since the actions performed by these nodes have not changed from $(V^{(M)},E^{(M)})$ to $(V,E)$. Therefore, the new code on $(V,E)$ achieving rate $R/M$ for the limited-node problem.

Now we show that $C\le\frac{1}{M-2z}C^{(M)}$. Take any code on $(V,E)$ achieving rate $R$. We will construct a code on $(V^{(M)},E^{(M)})$ achieving rate $(M-2z)R$. This direction is slightly more difficult because the new code needs to handle a greater variety of traitors. The code on $(V^{(M)},E^{(M)})$ is composed of an outer code and $M$ copies of the $(V,E)$ code running in parallel. The outer code is a $(M,M-2z)$ MDS code with coded output values $w_1,\ldots,w_M$. These values form the messages for the inner codes. Since we use an MDS code, if $w_1,\ldots,w_M$ are reconstructed at the destination such that no more than $z$ are corrupted, the errors can be entirely corrected. The $j$th copy of the $(V,E)$ code is performed by $i^*$ for $i\in U$, and by $i^{(j)}$ for $i\notin U$. That is, nodes in $U$ are each involved in all $M$ copies of the code, while nodes not in $U$ are involved in only one. Because the $(V,E)$ code is assumed to defeat any attack on only nodes in $U$, if for some $j$, no nodes $i^{(j)}$ for $i\notin U$ are traitors, then the message $w_j$ will be recovered correctly at the destination. Therefore, one of the $w_j$ could be corrupted only if $i^{(j)}$ is a traitor for some $i\notin U$. Since there are at most $z$ traitors, at most of the $w_1,\ldots,w_M$ will be corrupted, so the outer code corrects the errors.

From \eqref{eq:limitedbounds}, \eqref{eq:limitedlimit} is immediate. We can easily identify $M$ large enough to compute $C$ to any desired precision.
\end{proof}

The significance of Theorem~\ref{thm:limited} is that if we could calculate the capacity of any all-node problem, we could use \eqref{eq:limitedbounds} to calculate the capacity of any limited-node problem. Furthermore, it is easy to see that for large $M$ the cut-set bound of $(V^{(M)},E^{(M)})$ is simply $M$ times the cut-set bound of $(V,E)$. Hence a limited-node example with capacity less than the cut-set bound---such as the one in Section~\ref{subsec:boundexample}---leads directly to an all-node example with capacity less than the cut-set bound.

\subsection{Example with Capacity Less than Cut-Set}\label{subsec:boundexample}

\begin{figure}
\centerline{\begin{psfrags}\small
\psfrag{s}[c]{$S$}
\psfrag{1}[c]{$1$}
\psfrag{2}[c]{$2$}
\psfrag{3}[c]{$3$}
\psfrag{4}[c]{$4$}
\psfrag{5}[c]{$5$}
\psfrag{6}[c]{$6$}
\psfrag{7}[c]{$7$}
\psfrag{8}[c]{$8$}
\psfrag{9}[c]{$9$}
\psfrag{10}[c]{$10$}
\psfrag{d}[c]{$D$}
\psfrag{15}[c]{$$}
\psfrag{26}[c]{$$}
\psfrag{37}[c]{$$}
\psfrag{48}[c]{$$}
\psfrag{9d}[c]{$$}
\psfrag{10d}[c]{$$}
\includegraphics[scale=.6]{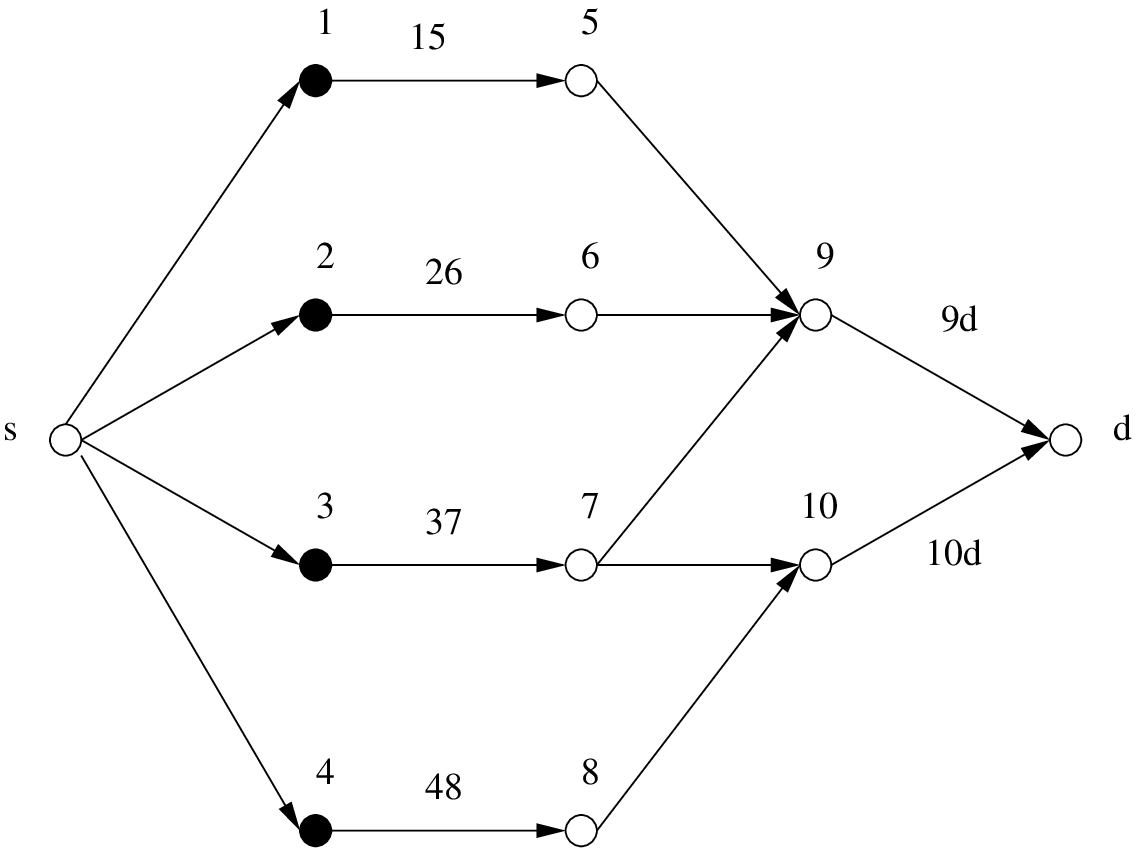}
\end{psfrags}
\hspace{.5in}
\includegraphics[scale=.6]{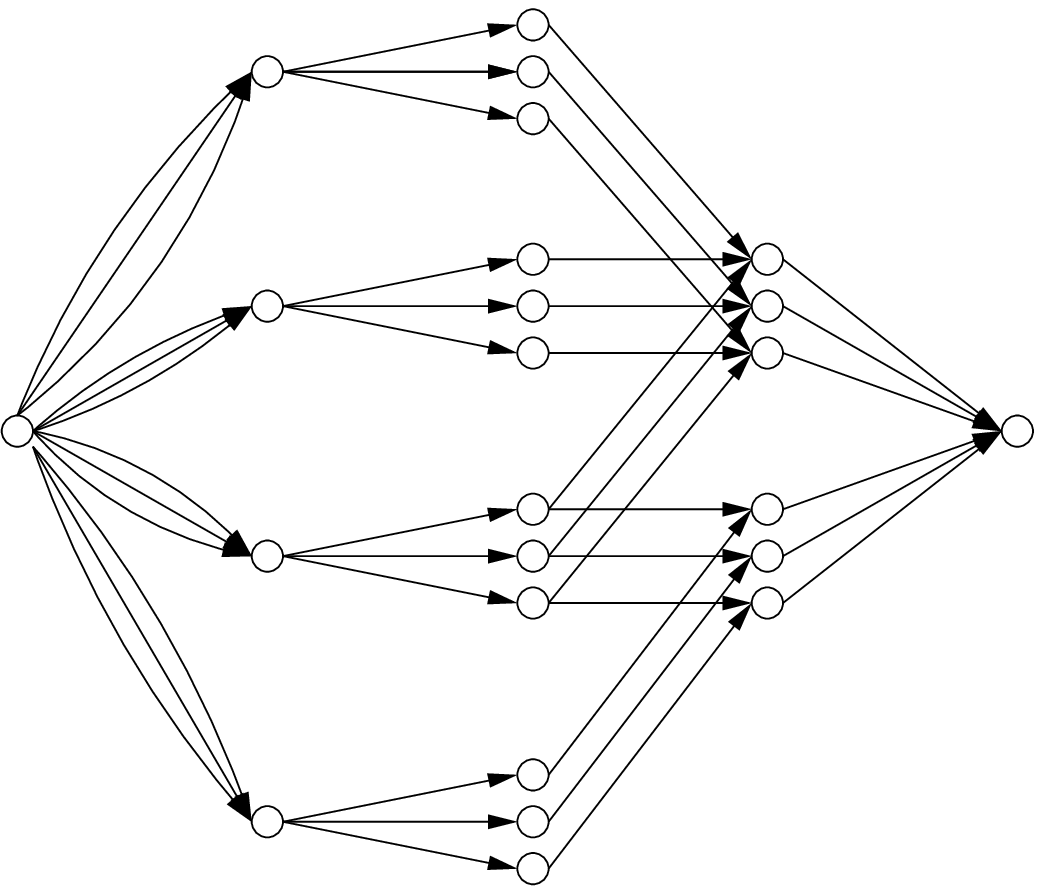}
}
\caption{A network with capacity strictly less than the cut-set bound. The limited-node network is shown on the left, and the equivalent all-node problem with 3 copies is shown on the right.}
\label{fig:42net}
\end{figure}

Consider the network shown in Figure~\ref{fig:42net}. There is at most one traitor, but it may only be one of nodes 1--4. The cut-set bound is easily seen to be 2, but in fact the capacity is no more than 1.5.

Suppose we are given a code achieving rate $R$. We show that $R\le 1.5$. For $i=1,2,3,4$, let $X_i$ be the random variable representing the value on the output edge of node $i$. Let $Y$ be the value on edge $(9,D)$ and let $Z$ be the value on $(10,D)$. Let $p$ be the honest distribution on these variables, and define the following alternative distributions:
\begin{align}
q_3&= p(x_1x_2x_4)p(x_3)p(y|x_1x_2x_3)p(z|x_3x_4),\\
q_4&= p(x_1x_2x_3)p(x_4)p(y|x_1x_2x_3)p(z|x_3x_4).
\end{align}
We may write
\beq R\le I_{q_3}(X_1X_2X_4;YZ)\label{eq:q3}\eeq
because, if node 3 is the traitor, it may generate a completely independent version of $X_3$ and send it along edge $(37)$, resulting in the distribution $q_3$. In that case, assuming the destination can decode properly, information about the message must get through from the honest edges at the start of the network, $X_1,X_2,X_4$, to what is received at the destination, $Y,Z$. From \eqref{eq:q3}, we may write
\begin{align}
R&\le I_{q_3}(X_1X_2X_4;Z)+I_{q_3}(X_1X_2X_4;Y|Z)\label{eq:q3.1}\\
&\le I_{q_3}(X_4;Z)+I(X_1X_2;Z|X_4)+1\label{eq:q3.2}\\
&=I_{q_3}(X_4;Z)+1\label{eq:q3.3}
\end{align}
where in \eqref{eq:q3.2} we have used that the capacity of $(9,D)$ is 1, and in \eqref{eq:q3.3} that $X_1X_2-X_4-Z$ is a Markov chain according to $q_3$. Using a similar argument in which node 4 is the traitor and it acts in a way to produce $q_4$, we may write
\beq R\le I_{q_4}(X_3;Z)+1.\label{eq:q4}\eeq
Note that
\beq q_3(x_3x_4z)=q_4(x_3x_4z).\eeq
In particular, the mutual informations in \eqref{eq:q3.3} and \eqref{eq:q4} can both be written with respect to the same distribution. Therefore,
\begin{align}
2R&\le I_{q_3}(X_4;Z)+I_{q_3}(X_3;Z)+2\label{eq:qr.1}\\
&=I_{q_3}(X_3X_4;Z)+I_{q_3}(X_3;X_4)-I_{q_3}(X_3;X_4|Z)+2\label{eq:qr.2}\\
&\le I_{q_3}(X_3X_4;Z)+2\label{eq:qr.3}\\
&\le 3\label{eq:qr.4}
\end{align}
where \eqref{eq:qr.3} follows from the positivity of conditional mutual information and that $X_3,X_4$ are independent according to $q_3$, and \eqref{eq:qr.4} follows because the capacity of $(10,D)$ is 1. Therefore, $R\le 1.5$.

Observe that all inequalities used in this upper bound were so-called Shannon-type inequalities. For the non-Byzantine problem, there is a straightforward procedure to write down all the Shannon-type inequalities relevant to a particular network coding problem, which in principle can be used to find an upper bound. This upper bound is more general than any cut-set upper bound, and in some multi-source problems it has been shown to be tighter than any cut-set bound. This example illustrates that a similar phenomenon occurs in the Byzantine problem even for a single source and single destination. As the Byzantine problem seems to have much in common with the multi-source non-Byzantine problem, it would be worthwhile to formulate the tightest possible upper bound using only Shannon-type inequalities. However, it is yet unclear what the ``complete'' list of Shannon type inequalities would be for the Byzantine problem. This example certainly demonstrates one method of finding them, but whether there are fundamentally different methods to find inequalities that could still be called Shannon-type, or even how to compile all inequalities using this method, is unclear. Moreover, it has been shown in the non-Byzantine problem that there can be active non-Shannon-type inequalities. It is therefore conceivable that non-Shannon-type inequalities could be active even for a single source under Byzantine attack.

\section{Conclusion}\label{sec:conclusion}

The main contribution of this paper has been to introduce the theory of Polytope Codes. As far as we know, they are the best known coding strategy to defeat generalized Byzantine attacks on network coding. However, it remains difficult to calculate the best possible rate they can achieve for a given network. We have proved that they achieve the cut-set bound, and hence the capacity, for a class of planar graphs, and we conjecture that this holds for all planar graphs. One would obviously hope to find the capacity of all networks, including non-planar ones. We have shown that achieving the cut-set bound is not always possible, meaning there remains significant work to do on upper bounds as well as achievable schemes. Whether Polytope Codes can achieve capacity on all networks remains an important open question.

\appendices

\section{Tighter Cut-set Upper bound}\label{appendix:cutset}

\begin{theorem}
Consider a cut $A\subseteq V$ with $S\in A$ and $D\notin A$. Let $E_A$ be the set of edges that cross the cut. For two not necessarily disjoint sets of possible traitors $T_1,T_2$, let $E_{1}$ and $E_{2}$ be the subset of edges in $E_A$ that originate at nodes in $T_1$ and $T_2$ respectively. Let $\tilde{E}$ be the set of edges in $E_1\cap E_2$ in addition to all edges $e\in E_1\cup E_2$ for which there is no path that flows through $e$ followed by any edge in $E_A\setminus E_1\setminus E_2$. The following upper bound holds on the capacity of the network:
\beq C\le \sum_{e\in E_A\setminus \tilde{E}} c_e.\label{eq:cutsettight}\eeq
\end{theorem}
\begin{proof}
Suppose \eqref{eq:cutset} were not true for some $A$, $T_1$, and $T_2$. Then there would exist a code achieving a rate $R$ such that
\beq R >\sum_{e\in E_A\setminus \tilde{E}} c_e.\label{eq:contra}\eeq
We will consider two possibilities, one when $T_1$ are the traitors and they alter the values on $E_1\cap\tilde{E}$, and one when $T_2$ are the traitors and they alter the values on $E_2\cap\tilde{E}$. We will show that by \eqref{eq:contra}, it is possible for the traitors to act in such a way that even though the messages at the source are different, all values sent across the cut are the same; therefore the destination will not be able to distinguish all messages. Note that traitors in $T_1$ or $T_2$ will only corrupt values on edges in $\tilde{E}$; that is, those edges controlled by either set of traitors, or those that could not influence $E_A\setminus E_1\setminus E_2$.

Fix a value $z$ representing one possible set of values that may be placed on the edges $E_1\cap E_2$. By definition, no edges in $\tilde{E}\setminus(E_1\cap E_2)$ are upstream of edges in $E_A\setminus\tilde{E}$. Since the traitors act honestly on all edges not in $\tilde{E}$, given $z$, the values on $E_A\setminus\tilde{E}$ are a function of the message, so by \eqref{eq:contra}, there exist two messages $w_a$ and $w_b$ that cannot be distinguished just from $E_A\setminus\tilde{E}$. Call $y$ the set of values on these edges under $w_a$ or $w_b$.

Choose a coding order on the edges in $\tilde{E}\setminus(E_1\cap E_2)$ written as
\beq (l_1,l_2,\ldots,l_K)\label{eq:codeorder}\eeq
where $K=|\tilde{E}\setminus(E_1\cap E_2)|$, such that if there is a path through $l_i$ followed by $l_j$, then $i<j$. Observe that $\tilde{E}\setminus(E_1\cap E_2)$ can be divided into $\tilde{E}\setminus E_2$ and $\tilde{E}\setminus E_1$, and therefore the order in \eqref{eq:codeorder} must alternate between edges in the two sets. Maintaining the order in \eqref{eq:codeorder}, we may group the edges by the two sets, rewriting \eqref{eq:codeorder} as
\beq (\Umsc_1,\Vmsc_1,\Umsc_2,\Vmsc_2,\ldots,\Umsc_{K'},\Vmsc_{K'})\eeq
where $\Umsc_i\subset \tilde{E}\setminus E_2$ and $\Vmsc_i\subset \tilde{E}\setminus E_2$, and $K'$ is the number of times the edges in \eqref{eq:codeorder} alternate between the two sets. Note that $\Umsc_1$ or $\Vmsc_{K'}$ may be empty.

We now construct the manner in which the two possible sets of traitors, $T_1$ or $T_2$, may cause $w_a$ and $w_b$ to become indistinguishable. Suppose $w_a$ is the message, $T_1$ are the traitors, they place $z$ on $E_1\cap E_2$, but behave honestly on all other edges. We denote the values on all edges crossing the cut as
\beq(z,y,u_1^{(1)},v_1^{(1)},u_2^{(1)},v_2^{(1)},\ldots,u_{K'}^{(1)},v_{K'}^{(1)})\label{eq:list1}\eeq
where $z$ represents the values on $E_1\cap E_2$, $y$ the values on $E_A\setminus\tilde{E}$, and $u_i^{(1)}$ and $v_i^{(1)}$ are the values placed on $\Umsc_i$ and $\Vmsc_i$ respectively. Note that only the $u_i$ values are directly adjustable by the traitors, but they may affect elements later in the sequence.

Alternatively, if $w_a$ is the message, $T_2$ are the traitors, and they place $z$ on $E_1\cap E_2$, but behave honestly elsewhere, the values across the cut are denoted
\beq(z,y,u_1^{(2)},v_1^{(2)},u_2^{(2)},v_2^{(2)},\ldots,v_{K'}^{(2)}).\label{eq:list2}\eeq
Here, only the $v_i$ values may be changed directly by the traitors.

In the two scenarios, the traitors may alter their output values so that \eqref{eq:list1} and \eqref{eq:list2} are transformed to become identical. This can be done as follows. In \eqref{eq:list1}, the traitors may replace $u_1^{(1)}$ with $u_1^{(2)}$. Downstream edges are either controlled by honest nodes, or they are controlled by traitors that continue, for now, to behave honestly. Hence, this change may affect the later edges in the sequence, but they do so in a way determined only by the code. This results in a set of values denoted by
\beq(z,y,u_1^{(2)},v_1^{(3)},u_2^{(3)},v_2^{(3)},\ldots,u_{K'}^{(3)},v_{K'}^{(3)}).\label{eq:list3}\eeq
In \eqref{eq:list2}, the traitors may now replace $v_1^{(2)}$ with $v_1^{(3)}$, resulting in
\beq(z,y,u_1^{(2)},v_1^{(3)},u_2^{(4)},v_2^{(4)},\ldots,v_{K'}^{(4)}).\eeq
We may now return to \eqref{eq:list3} and replace $u_2^{(3)}$ with $u_2^{(4)}$, further changing downstream values. Continuing this process will cause the two sequences to become identical, thereby making $w_a$ and $w_b$ indistinguishable at the destination.
\end{proof}

\section{Proof of Bound on Linear Capacity for the Cockroach Network}\label{appendix:linear}

We show that no linear code for the Cockroach Network, shown in Figure~\ref{fig:cock}, can achieve a rate higher than 4/3. Fix any linear code. For any link $(i,j)$, let $X_{i,j}$ be the value placed on this link. For every node $i$, let $X_i$ be the set of messages on all links out of node $i$, and $Y_i$ be the set of messages on all links into node $i$. Let $G_{X_i\to Y_j}$ be the linear transformation from $X_i$ to $Y_j$, assuming all nodes behave honestly. Observe that
\beq Y_D=G_{X_S\to Y_D}X_S(w)+\sum_i G_{X_i\to Y_D}e_i\eeq
where $e_i$ represents the difference between what a traitor places on its outgoing links and what it would have placed on those links if it were honest. Only one node is a traitor, so at most one of the $e_i$ is nonzero. Note also that the output values of the source $X_S$ is a function of the message $w$. We claim that for any achievable rate $R$,
\beq R\le \frac{1}{n}\left[\rank(G_{X_S\to Y_D})-\max_{i,j}\rank(G_{X_iX_j\to Y_D})\right]\label{eq:rankbound}\eeq
where $n$ is the block length used by this code. To show this, first note that for any pair of nodes $i,j$ there exist $K,H_1,H_2$ such that
\beq G_{X_S\to Y_D}=K+G_{X_i\to Y_D}H_1+G_{X_j\to Y_D}H_2\label{eq:gdecomp}\eeq
and where
\beq\rank(K)=\rank(G_{X_S\to Y_D})-\rank(G_{X_iX_j\to Y_D}).\eeq
That is, the first term on the right hand side of \eqref{eq:gdecomp} represents the part of the transformation from $X_S$ to $Y_D$ that cannot be influenced by $X_i$ or $X_j$. Consider the case that $\rank(K)<R$. Then there must be two messages $w_1,w_2$ such that $KX_S(w_1)=KX_S(w_2)$. If the message is $w_1$, node $i$ may be the traitor and set
\beq e_i=H_1(X_S(w_2)-X_S(w_1)).\eeq
Alternatively, if the message is $w_2$, node $j$ may be the traitor and set
\beq e_j=H_2(X_S(w_1)-X_S(w_2)).\eeq
In either case, the value received at the destination is
\begin{multline*}
Y_D=KX_S(w_1)+G_{X_i\to Y_D}H_1X_S(w_2)\\+G_{X_j\to Y_D}H_2X_S(w_1).\end{multline*}
Therefore, these two cases are indistinguishable to the destination, so it must make an error for at least one of them. This proves \eqref{eq:rankbound}.

Now we return to the specific case of the Cockroach Network. Observe that the $X_{4,D}$ is a linear combination of $X_{1,4}$ and $X_{2,4}$. Let $k_1$ be the number of dimensions of $X_{4,D}$ that depend only on $X_{1,4}$ and are independent of $X_{2,4}$. Let $k_2$ be the number of dimensions of $X_{4,D}$ that depend only on $X_{2,4}$, and let $k_3$ be the number of dimensions that depend on both $X_{1,4}$ and $X_{2,4}$. Certainly $k_1+k_2+k_3\le n$. Similarly, let $l_1,l_2,l_3$ be the number of dimensions of $X_{5,D}$ that depend only on $X_{2,5}$, that depend only on $X_{3,5}$, and that depend on both respectively. Finally, let $m_1$ and $m_2$ be the number of dimensions of $X_{1,D}$ and $X_{3,D}$ respectively.

We may write the following:
\begin{align*}\rank(G_{X_S\to Y_4})-\rank(G_{X_2,X_3\to Y_4})&\le m_1+k_1,\\
\rank(G_{X_S\to Y_4})-\rank(G_{X_1,X_3\to Y_4})&\le k_3+l_1,\\
\rank(G_{X_S\to Y_4})-\rank(G_{X_1,X_2\to Y_4})&\le l_3+m_2.
\end{align*}
Therefore, using \eqref{eq:rankbound}, any achievable rate $R$ is bounded by
\beq R\le \frac{1}{n}\min\{m_1+k_1,k_3+l_1,l_3+m_2\}\eeq
subject to
\begin{align}
k_1+k_2+k_3&\le n,\\
l_1+l_2+l_3&\le n,\\
m_1\le n,\\
m_2\le n.\end{align}
It is not hard to show that this implies $R\le 4/3$.


\end{document}